\ifx\onecol\undefined
\documentclass[journal,twocolumn]{IEEEtran}
\else 
\documentclass[12pt,journal,onecolumn]{IEEEtran}
\fi

\usepackage{amsmath,amsfonts}
\usepackage{algorithmic}
\usepackage{algorithm}
\usepackage{array}
\usepackage{color}
\usepackage{cuted}
\usepackage{diagbox}
\usepackage{textcomp}
\usepackage{stfloats}
\usepackage{url}
\usepackage{verbatim}
\usepackage{graphicx}
\graphicspath{{Figures/}}
\usepackage{cite}
\usepackage{amsthm}
\usepackage{amssymb}
\usepackage{gensymb}
\usepackage{bm}
\usepackage{setspace}
\usepackage{tabularx}
\usepackage{hyperref}
\usepackage{subfigure}
\hypersetup{hidelinks, 
colorlinks=true,
allcolors=black,
pdfstartview=Fit,
breaklinks=true}

\newtheoremstyle{mytheoremstyle}%
  {\topsep}
  {\topsep}
  {\normalfont}
  {1em}
  {\itshape}
  {:}
  {.5em}
  {\thmname{#1}\thmnumber{ #2}\thmnote{ (#3)}}

\theoremstyle{mytheoremstyle}
\newtheorem{lemma}{Lemma}

\renewenvironment{IEEEproof}[1][\IEEEproofname]
  {\par\noindent\hspace{\IEEEproofindentspace}{\itshape #1: }}
  {\par}

\begin{document}

\title{Near-Field Communications with Grating Lobes \\ for Quasi-Distributed Arrays: From ULA to MRA}

\author{Mingyuan Zhou, Zhuo Xu, and Linglong Dai, \textit{Fellow, IEEE} \vspace{-2em}

\thanks{This work is supported in part by the National Science and Technology Major Projects of China under Grant 2025ZD1301800, in part by the National Science Fund for Distinguished Young Scholars under Grant 62325106, and in part by the National Key Research and Development Program of China under Grant 2023YFB3811503.

M. Zhou, Z. Xu, and L. Dai are with the Department of Electronic Engineering, Tsinghua University, Beijing 100084, China, and also with the State Key Laboratory of Space Network and Communications, Tsinghua University, Beijing 100084, China (e-mails: \href{mailto:zhoumy22@mail.tsinghua.edu.cn}{zhoumy22@mails.tsinghua.edu.cn}; \href{mailto:xz23@tsinghua.org.cn}{xz23@tsinghua.org.cn}; \href{mailto:daill@tsinghua.edu.cn}{daill@tsinghua.edu.cn}).
}
}



\maketitle

\begin{abstract}

Extremely large-scale antenna array (ELAA) has emerged as a common feature of many key candidate technologies for 6G, where the near-field characteristics become dominant. 
The quasi-distributed array can further extend the near-field range and 
utilize the near-field benefits to improve the system performance.
However, its typical implementation with modular arrays suffers from severe grating lobes that cause non-negligible inter-user interferences. 
To solve this problem, we propose the modular minimum-redundancy array (M-MRA) to suppress near-field grating lobes by redesigning the subarray configuration. 
Specifically, we first characterize the beam pattern of the conventional modular uniform linear array (M-ULA).
Contrary to the common belief that grating lobes only exist in the angle domain, we reveal that near-field grating lobes may also occur in the distance domain. 
We further analyze how to suppress near-field grating lobes for the M-ULA. 
The results demonstrate that increasing the number of antennas per module can suppress grating lobes.
In particular, the number of antennas required grows linearly with the inter-module spacing, thus the grating lobe interferences are severe under a limited number of antennas. 
This limitation inspires us to propose the M-MRA by redesigning the subarray structure. 
For M-MRA, the nonuniform antenna spacing within each subarray provides a narrower spatial envelope, allowing it to suppress near-field grating lobes in the angle and distance domains simultaneously. 
Simulation results verify that the proposed M-MRA can significantly improve the spectrum efficiency of multi-user near-field communications under the same number of antennas.

\end{abstract}


\section{Introduction}
\label{section introduction}

\IEEEPARstart{F}{uture} sixth-generation (6G) wireless networks are expected to support applications such as immersive interaction and digital twins, which requires higher data rate and spectrum efficiency \cite{saad2020vision,giordani2020toward}. To realize these visions, different key candidate air-interface technologies have been proposed to meet different key performance indicators (KPIs). For example, extremely large-scale multiple-input multiple-output (XL-MIMO) and cell-free MIMO can substantially improve the system spectrum efficiency. Reconfigurable intelligent surfaces (RIS) are considered to extend the wireless coverage, and terahertz (THz) communications can support higher peak data rates. All technologies discussed above share the deployment of extremely large-scale antenna array (ELAA) as a common physical feature \cite{wang2024tutorial}, \cite{ye2024elaa}, \cite{lu2024tutorial}. By accommodating hundreds of antennas, ELAA can provide more abundant spatial degrees of freedom (DoFs) and higher spatial resolution \cite{cui2023nearfield,an2024nearfield}.


The substantial increase in the number of antennas for the ELAA leads to a fundamental change in the electromagnetic propagation characteristics, resulting in a prominent expansion of the radiative near-field range of ELAA. 
Specifically, the boundary between the radiative near-field and far-field regions is commonly characterized by the Rayleigh distance, i.e.,  $\frac{2D^2}{\lambda}$, where $D$ is the array aperture and $\lambda$ is the carrier wavelength \cite{selvan2017fraunhofer}. Consequently, as the array aperture of ELAA increases, a considerable portion of the intended users may fall within the radiative near field, where the phase variation over the array is governed by spherical wavefronts, rather than the far-field planar wavefronts \cite{bjornson2021primer,zhang2023farfield}. 

Based on the near-field\footnote{The reactive near-field range of a single antenna is only on the wavelength level. Besides, the electromagnetic energy inside the reactive near-field is in an oscillating state and will not radiate outward. Thus, ``near field'' refers to the radiative near-field region throughout this paper.} spherical wavefronts, there emerge three distinct features that can be exploited to improve the system performance. 
Firstly, unlike the far-field line-of-sight (LoS) channel with only one available data stream, the near-field LoS channel becomes highly ranked \cite{zidong2022DoF}. The DoFs are increased accordingly, which is expected to enhance the channel capacity for a single user. 
Secondly, different from far-field beamforming, near-field beamfocusing can focus energy at specific locations, allowing users at identical angle but different distances to be served simultaneously \cite{zhang2022beamfocusing}. By exploiting the additional near-field resources in the distance domain, location division multiple access (LDMA) is expected to improve the spectrum efficiency for multi-user near-field communications \cite{wu2023ldma}.
Thirdly, special near-field beams, such as Airy beam, can circumvent obstacles \cite{Airy2024}, which is promising for addressing the blockage problem.

\subsection{Prior Works}
\label{subsec:prior_works}

The extent to which these near-field benefits can be exploited is closely related to the physical aperture. 
A conventional collocated array supports centralized baseband processing, but has a limited aperture and hence limited near-field range under a fixed number of antennas. In contrast, a distributed array deploys multiple geographically separated access points. The distributed deployment can enlarge the effective array aperture and extend the near-field range \cite{ngo2017cellfree}, \cite{bjornson2020cellfree}, \cite{ngo2024ultradense}, but requires extensive exchange of channel information and strict timing calibration across multiple access points \cite{ammar2022distributed}, \cite{rogalin2014sync}, \cite{balan2013airsync}. As a compromise between the collocated and distributed arrays, the \textbf{quasi-distributed array} spatially separates multiple antenna panels or subarrays within one base station. In this way, the quasi-distributed array can extend the near-field range while retaining the benefits of centralized baseband processing. 


One of the typical implementation forms of the quasi-distributed array is the modular array, where multiple subarrays within one base station are separated by enlarged inter-module spacing \cite{jeon2021modular,li2022modularmodel}. 
By this means, the near-field range can be expanded while the precision requirement on the inter-module synchronization can be loosen. 
Thus, modular arrays have attracted growing research interest in near-field communications. 
Specifically, the concept of modular deployment and its implementation motivation were discussed in \cite{jeon2021modular}. 
On this basis, the spherical-wave channel model and the signal-to-noise-ratio (SNR) behavior were developed for modular arrays in \cite{li2022modularmodel}. 
To further characterize the spatial energy distribution, 
the near-field beam pattern for modular arrays was investigated in \cite{li2023grating}. 
Moreover, the beamwidth and beamdepth properties of the modular uniform linear array (M-ULA) with two subarrays were studied in \cite{kosasih2026modular}. 
These works have shown that, the modular array can translate the physical advantages into system-level gains. The enlarged near-field range and the higher near-field beamfocusing precision both contribute to enhancing the spectrum efficiency in multi-user near-field communication systems.

Unfortunately, for near-field modular arrays, if the modules are arranged periodically and the inter-module spacing is larger than half a wavelength, near-field grating lobes tend to emerge \cite{li2023grating,li2024multiuser}. Severe inter-user interferences are introduced accordingly, which may degrade the system performance. 
Existing works mainly suppress near-field grating lobes through algorithm designing or array configuration designing. 
At the software level, a user-grouping scheduling strategy was proposed in \cite{li2024multiuser}, which ensures that users severely affected by gating lobe interferences are not scheduled on the same time-frequency resource block. Since the user-grouping method needs to determine the complete near-field channel state information (CSI) in advance, it only performs well in quasi-static scenarios but fail in general non-static scenarios. 
To essentially suppress near-field grating lobes for modular arrays, methods through changing the array configuration are further studied. For example, extended coprime array (ECA) was investigated in \cite{zhou2025sparse}. By making the antenna spacing comprime, ECA causes an offset in the position of near-field grating lobes, thus reducing the height of grating lobes. In addition, tapered aperture sampling (TAS) was proposed in \cite{parvini2026tapered}, which utilized nonuniform sampling to redistribute the modules and break the inherent grating lobe structure. 
However, these hardware-based solutions to near-field grating lobe suppression have inherent limitations. Since ECA and TAS suppress grating lobes by changing the position of each module, they both require a sufficient number of modules. For the near-field channel model defined by the 3rd Generation Partnership Project (3GPP) in June, 2026 \cite{3GPP2026}, the number of modules along a certain direction is no more than 4. In this case, the performance of suppressing near-field grating lobes by changing the position of each module is almost negligible.

\subsection{Our Contributions}
\label{subsec:contributions}

To essentially suppress the near-field grating lobes for general modular arrays, we propose the modular minimum-redundancy array (M-MRA) by redesigning the configuration within each single module\footnote{Simulation codes will be provided to reproduce the results in this article: http://oa.ee.tsinghua.edu.cn/dailinglong/publications/publications.html.}. 
The main contributions are summarized as follows:

\begin{itemize}
    \item \textbf{Near-field beam pattern analysis for M-ULA with multiple subarrays:}
    We first characterize the near-field beam pattern of an M-ULA with an arbitrary number of subarrays. In contrast to the common belief that near-field grating lobes only exist in the angle domain, we reveal that near-field grating lobes may also occur in the distance domain. This phenomenon is consistent with the additional distance-domain resource brought by near-field spherical wavefronts. Besides, we discover that the beam patterns for modular arrays with even or odd number of modules are different. For both two cases, closed-form expressions for the locations, beamwidth, and beamdepth of the grating lobes are derived.

    \item \textbf{Near-field grating lobe suppression for M-ULA:}
    Based on the near-field beam pattern above, we further analyze how to suppress near-field grating lobes for M-ULA. Our analysis demonstrates that, increasing the number of antennas per module can suppress near-field grating lobes. Specifically, the number of antennas per module should grow linearly with the inter-module spacing. Beyond the common belief that the inter-module spacing can be enlarged blindly to expand the near-field range, we point out the limitation. To avoid severe grating lobe interferences for an M-ULA, if the inter-module spacing is to be increased, each subarray should be equipped with a large number of antennas.

    \item \textbf{Antenna-efficient near-field grating lobe suppression by M-MRA:}
    To overcome the limitation on the number of antennas required, we redesign the subarray configuration to suppress near-field grating lobes. Specifically, we propose the M-MRA structure. Unlike the M-ULA, whose antenna spacing within each subarray is uniform, each module of M-MRA follows a minimum-redundancy placement. The nonuniform intra-module spacing of M-MRA provides an envelope with a narrower main lobe, thereby attenuating grating lobes in both the angle and distance domains. Simulation results verify that, compared to the conventional M-ULA, proposed M-MRA can substantially suppress near-field grating lobes and improve the near-field multi-user spectrum efficiency under the same number of antennas.
\end{itemize}

\subsection{Organization and Notation}
\label{subsec:organization}

The remainder of this paper is organized as follows. Section~II presents the system model. Section~III presents the near-field array gain, based on which we further analyze the angle-distance two-dimensional beam pattern of M-ULA with multiple subarrays. In Section~IV, we explore how to suppress grating lobes for M-ULA and propose the M-MRA configuration. Simulation results are provided in Section~V, followed by the conclusions in Section~VI.

\emph{Notation:} Lower-case and upper-case boldface letters denote vectors and matrices, respectively. The transpose and conjugate transpose are represented by $(\cdot)^{T}$ and $(\cdot)^{H}$, respectively. The operators $|\cdot|$ denotes the absolute value, and the set of $M\times N$ complex matrices is denoted by $\mathbb{C}^{M\times N}$. The Gaussian distribution with mean $\bm{\mu}$ and covariance $\bm{\Sigma}$ is denoted by $\mathcal{CN} \left(\bm{\mu},\bm{\Sigma}\right)$.

\section{System Model}
\label{section system model}

In this section, the array model for the quasi-distributed array is first introduced. Then, the near-field channel model and signal model are presented.


\subsection{Array Model}

\begin{figure*}[t!]
    \centering 
    \subfigure[]{
    \raisebox{0.2\height}{\includegraphics[width=0.21 \linewidth]{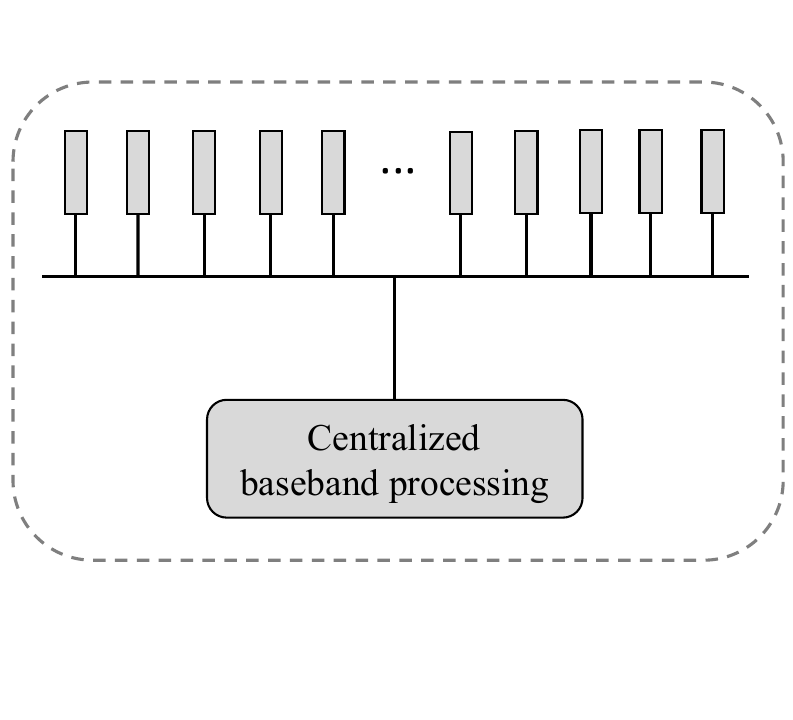}}}
    \hfill
    \subfigure[]{
    \includegraphics[width=0.26 \linewidth]{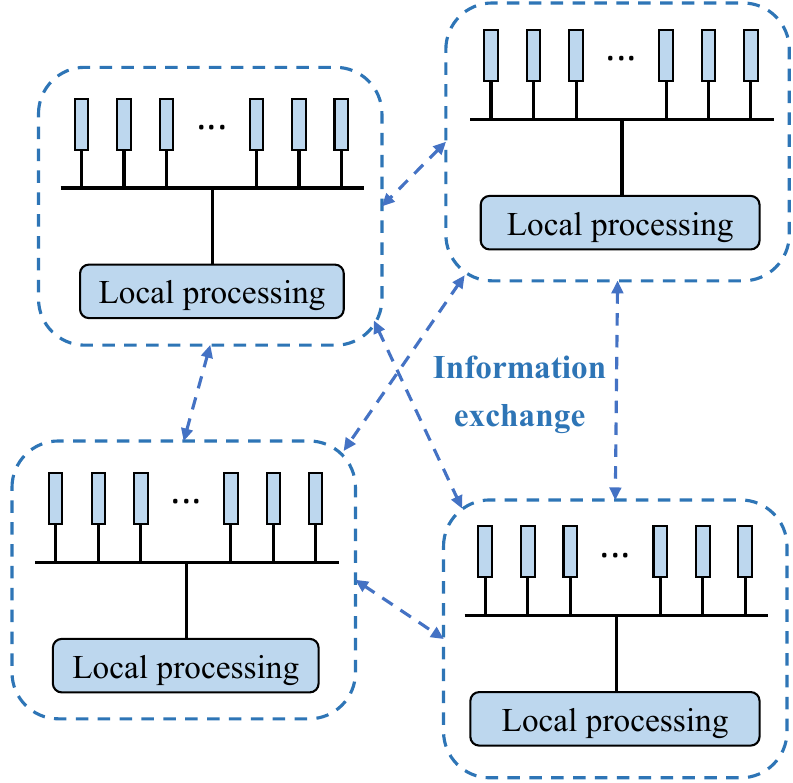}}
    \hfill
    \subfigure[]{
    \includegraphics[width=0.47 \linewidth]{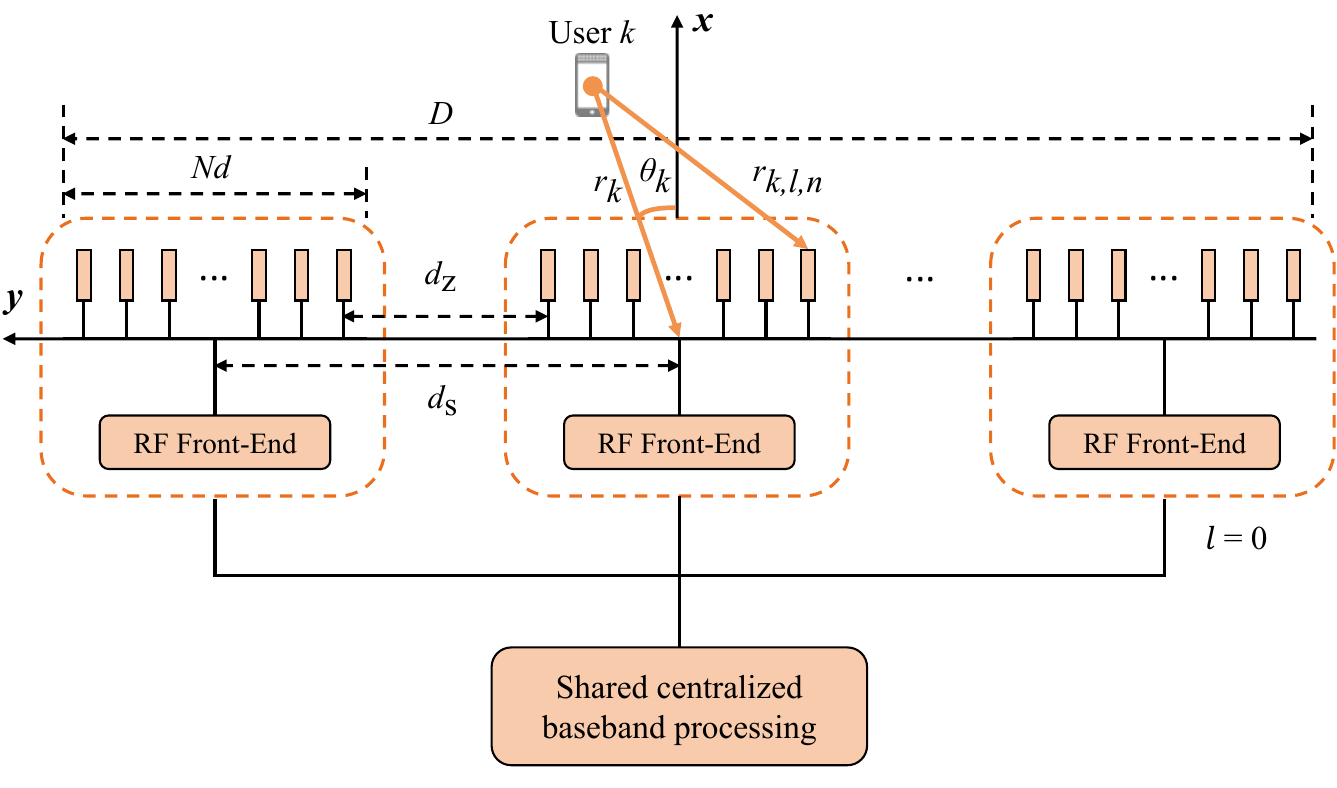}}
    \caption{Illustration of three types of arrays: (a) Collocated array; (b) Distributed array; (c) Quasi-distributed array.}
    \label{fig:System Model}
\end{figure*}

The collocated array, distributed array and quasi-distributed array are illustrated in Fig.~\ref{fig:System Model} respectively. As a compromise between the collocated and distributed arrays, the quasi-distributed array deploys multiple spatially separated antenna panels or subarrays within one base station. In this way, a quasi-distributed array can extend its array aperture and near-field range, exploiting more resources in the near-field region. Meanwhile, its centralized baseband processing can loosen the requirements on inter-module synchronization. One of the typical implementation forms of the quasi-distributed array is the modular array.

We consider a downlink wireless communication system shown as Fig.~\ref{fig:System Model} (c), where a modular array is deployed at the base station to serve $K$ single-antenna users. We denote $L$ as the number of modules and $N$ the number of antennas within each module. The total number of antennas is $LN$ accordingly. The modular array is assumed to be linearly arranged along the \textit{y}-axis, with $d_{\mathrm{s}}$ denoting the distance between the centers of two adjacent modules. With $l \in \mathcal{L} = \{ 0, \cdots, L-1 \} $ representing the module index and $n \in \mathcal{N} = \{ 0, \cdots, N-1 \} $ the antenna index within each module, we denote $Y_l$, $y_n$ and $y_{l,n} = Y_{l} + y_{n}$ as the \textit{y}-coordinate of the center of the \textit{l}-th module, the offset of the \textit{n}-th antenna within each module relative to the module center, and the absolute \textit{y}-coordinate of the \textit{n}-th antenna in the \textit{l}-th module, respectively. In addition, we denote $f$ as the carrier frequency, $\lambda$ the carrier wavelength, $d=\frac{\lambda}{2}$ the basic unit of antenna spacing, $D_{\mathrm{sub}}$ the aperture of each module, and $D$ the aperture of the whole modular array.

Specifically, if each module is arranged as a uniform linear array (ULA) with the inter-antenna spacing being exactly $d$, $y_n$ could be expressed as $y_{n} = (n-\frac{N-1}{2})d$. We denote $d_{\mathrm{z}}$ as the distance between two closest antennas of two adjacent modules, then $d_{\mathrm{s}}$, $Y_{l}$ and $D$ could be expressed as $d_{\mathrm{s}}=d_{\mathrm{z}}+(N-1)d$, $Y_{l} = (l-\frac{L-1}{2})d_{\mathrm{s}}$ and $D=(L-1)d_{\mathrm{s}}+Nd$.


\subsection{Near-Field Channel Model}

We suppose that user \textit{k} is located at $(x_k,y_k)=(r_k\cos\theta_k, r_k\sin\theta_k)$, where $k \in \mathcal{K} = \{ 0, \cdots, K-1 \} $ is the user index, $r_k$ is the the distance from user \textit{k} to the array center, and $\theta_k \in [-\frac{\pi}{2},\frac{\pi}{2}]$ is the angle with respect to the positive \textit{x}-axis. As has been discussed in Section \ref{section introduction}, the modular array 
substantially increases the effective array aperture, which in turn enlarges the near-field range, making it reasonable to assume that all users are located in the near-field region for the modular array, i.e., $\forall k\in\mathcal{K}$, $r_k\leq \frac{2D^2}{\lambda}$, where $\frac{2D^2}{\lambda}$ is the boundary between the far-field and near-field\cite{selvan2017fraunhofer}. 

Under this near-field condition, the planar wavefront assumption used in far-field communications is no longer accurate. Instead, the spherical wavefront propagation model must be adopted \cite{cui2023nearfield}. For near-field communications in high-frequency bands, non-line-of-sight (NLoS) paths suffer severe path loss and the LoS path dominates the channel propagation\cite{NoNLoS}. Based on the near-field spherical-wave model, the near-field channel from the modular array to user \textit{k}, i.e., $\mathbf{h}_k^\mathrm{near} \in\mathbb{C}^{(\textit{LN}) \times1}$, can be expressed as \cite{Cui2022channel}
\begin{equation}
  \label{h}
  \mathbf{h}_k^\mathrm{near}= \sqrt{LN}\alpha_0\mathbf{b}(r_k,\theta_k),
\end{equation}
where $\alpha_0$ denotes the complex gain of the LoS path and $\mathbf{b}(r_k,\theta_k) \in\mathbb{C}^{(\textit{LN})\times1}$ denotes the near-field steering vector for user \textit{k}. The $((\textit{l}-1)N+n)$-th element of $\mathbf{b}(r_k,\theta_k)$ can be expressed as
\begin{equation}
  \label{eq:NF-steering-vector}
  \left[\textbf{b}(r_k, \theta_k)\right]_{l,n} = \frac{1}{\sqrt{LN}} e^{-j\frac{2\pi}{\lambda}\left(r_{k,l,n}-r_k\right)},
\end{equation}
where $r_{k,l,n}$ denotes the distance between user \textit{k} and the \textit{n}-th antenna in the \textit{l}-th module. $r_{k,l,n}$ could be expressed as
\begin{equation}
\label{eq:r-precise}
\begin{aligned}
  r_{k,l,n} &= \sqrt{(r_k\text{cos}\theta_k)^2 + (r_k\text{sin}\theta_k - y_{l,n})^2}\\
  &= \sqrt{r_k^2 + y_{l,n}^2 - 2r_ky_{l,n}\text{sin}\theta_k}\\
  &\overset{(a)}{\approx}r_k - y_{l,n}\sin\theta_k + \frac{y_{l,n}^2\cos^2\theta_k}{2r_k},\\
\end{aligned}  
\end{equation}
where approximation (a) is the Fresnel approximation, which is derived by Taylor's Formula, i.e., $\sqrt{1+x}=1+\frac x2-\frac{x^2}8+\mathcal{O}(x^3)$\cite{Zhuo2025LLM}. Intuitively, this approximation retains the first-order and second-order terms of the distance expansion, which capture the angle and curvature information of the spherical wavefront, respectively. It has been proved that the approximation is accurate when $r_k$ is larger than $0.5\sqrt{\frac{D^3}{\lambda}}$, which is much smaller than the Rayleigh distance and is nearly negligible\cite{selvan2017fraunhofer}.

\subsection{Signal Model}
For the downlink communication system discussed above, we denote $\mathbf{h}_{k}$ as the near-field channel for user \textit{k} for simplicity, then the received signal at user \textit{k} can be expressed as
\begin{equation}
    \label{eq:receivedsignal}
    y_k=\mathbf{h}_{k}^H\mathbf{W}\mathbf{P}\mathbf{s}+n_k,
\end{equation}
where $\mathbf{W}=[\mathbf{w}_1, \mathbf{w}_2, \cdots, \mathbf{w}_{K}] \in \mathbb{C} ^{(\textit{LN})\times \textit{K}}$ denotes the transmit precoding matrix, with $\mathbf{w}_{k}$ representing the near-field beamfocusing vector for user \textit{k}. $\mathbf{P} = \rm{diag} \{\sqrt{\textit{P}_1},\sqrt{\textit{P}_2},\cdots,\sqrt{\textit{P}_{\textit{K}}}\}\in\mathbb{C}^{\textit{K}\times \textit{K}}$ denotes the power allocation matrix satisfying $\textstyle\sum_{k=1}^{K}P_k\leq P$, with $P$ denoting the maximum transmit power. $\mathbf{s}$ denotes the transmitted signal, which has been power-normalized under the constraint $\mathbb{E}[\mathbf{s}\mathbf{s}^H]=\mathbf{I}$, and $n_k \sim \mathcal{CN} \left(0,\sigma_k^{2}\right)$ denotes the additive white Gaussian noise (AWGN) with power $\sigma_k^{2}$. Based on the above, the received signal-to-interference-plus-noise ratio (SINR) at user \textit{k} can be expressed as
\begin{equation}
    \label{eq:SINR}
    \mathrm{SINR}_k = \frac{P_k \left|\mathbf{h}_k^{H} \mathbf{w}_k\right|^2}
    {\sum\limits_{k^{\prime}=1, k^{\prime}\neq k}^{K} P_{k^{\prime}} \left|\mathbf{h}_k^{H} \mathbf{w}_{k^{\prime}}\right|^2 + \sigma_k^2},
\end{equation}
and the corresponding spectrum efficiency upper bound can be expressed as
\begin{equation}
    \label{eq:SE}
    \mathrm{SE}_k = \mathrm{log}_2\left( 1+\mathrm{SINR}_k \right).
\end{equation}

So far, we have introduced the channel and signal models of modular array near-field communication systems, based on which the analysis of near-field beam pattern for modular arrays will be explained in the following section.

\section{Analysis of Near-Field Beam Pattern \\ for M-ULA}
\label{section_analysis}

In Section~\ref{section system model}, we have discussed the near-field modular array models. 
In this section, we first present analysis of near-field beam pattern for modular uniform linear array (M-ULA), based on which we further analyze its near-field grating lobe locations and beamfocusing spatial extent. 

\subsection{Analysis of Near-Field Beam Pattern}

The key insight that facilitates our analysis is that the beam pattern of a modular array can be decomposed into a slowly-varying envelope mainly determined by the intra-module antenna arrangement, and a rapidly-varying array factor determined by the inter-module structure. As will be shown, the grating lobe locations are dictated primarily by the inter-module factor, while the exact power is modulated by the envelope. Specifically, for a beamfocusing vector $\mathbf{w}_\textrm{F} \in \mathbb{C} ^{(\textit{LN})\times 1}$ with $(r_\textrm{F},\theta_\textrm{F})$ being the desired beamfocusing location, the beam pattern characterizing the normalized beam power at an arbitrary observation location $(r,\theta)$ can be expressed as
\begin{equation}
    \label{eq:G}
    G(r,\theta;\mathbf{w}_\textrm{F}) \triangleq \left| \mathbf{w}_\textrm{F}^H \mathbf{b}(r,\theta) \right|^2,
\end{equation}
where $\mathbf{b}(r,\theta)$ denotes the near-field steering vector modeled by (\ref{eq:NF-steering-vector}). We consider the near-field beamfocusing based on the maximum ratio transmission (MRT), under which assumption we have $\mathbf{w}_\textrm{F} = \mathbf{b}(r_\textrm{F}, \theta_\textrm{F})$. By substituting $\mathbf{b}(r,\theta)$ in (\ref{eq:NF-steering-vector}) into (\ref{eq:G}), the near-field beam pattern can be expressed as
\begin{equation}
  \label{eq:G_MRT}
  \begin{aligned}
  & G(r,\theta;r_\textrm{F}, \theta_\textrm{F}) \triangleq \left| \mathbf{b}^H (r_\textrm{F}, \theta_\textrm{F}) \mathbf{b}(r,\theta) \right|^2 \\
  = &\left| \frac{1}{LN}\sum_{l = 0}^{L-1}\sum_{n = 0}^{N-1}
  e^{j\kappa y_{l,n} \left(\sin\theta - \sin\theta_\textrm{F}\right)}
  e^{j\kappa \frac{y^2_{l,n}}{2}\left( \frac{\cos^2\theta_\textrm{F}}{r_\textrm{F}} - \frac{\cos^2\theta}{r} \right)} \right|^2,
  \end{aligned}
\end{equation}
where $\kappa=\frac{2\pi}{\lambda}$ represents the wavenumber.

For convenience of notation, we denote $\Delta=\sin\theta - \sin \theta_\textrm{F}$ as the spatial angle difference, which describes the linear phase, and $\phi=\frac{\cos^2\theta_\textrm{F}}{r_\textrm{F}} - \frac{\cos^2\theta}{r}$ as the spherical-wavefront curvature difference, which describes the quadratic phase. Specifically, the location set $\left\{(r,\theta)| \frac{\cos^2\theta_\textrm{F}}{r_\textrm{F}} = \frac{\cos^2\theta}{r} \right\}$, or equivalently $\left\{(r,\theta)|\phi=0 \right\}$, is named as the \textit{user-ring}, among which the focal point $(r_\textrm{F},\theta_\textrm{F})$ is a representative location. By substituting $\Delta$, $\phi$ and $y_{l,n} = Y_{l} + y_{n}$ into (\ref{eq:G_MRT}), we can obtain
\begin{equation}
  \label{eq:G_MRT_noation}
  \begin{aligned}
  & G(r,\theta;r_\textrm{F}, \theta_\textrm{F}) \\ = &\left|\frac{1}{LN} \sum_{l=0}^{L-1} e^{j\kappa Y_{l}\Delta}e^{j\kappa\frac{Y_{l}^2}{2}\phi} \left( \sum_{n=0}^{N-1} e^{j\kappa y_{n}\Delta} e^{j\kappa\frac{y_{n}^2}{2}\phi} e^{j\kappa Y_{l}y_{n}\phi} \right) \right|^2.
  \end{aligned}
\end{equation}

Since the aperture of each module is rather small, its near-field range is limited accordingly. For example, for an M-ULA with $N=$\ 25 at $f=$\ 15 GHz, the near-field region for each ULA is restricted within the distance of 6.25 m. Therefore, users could be considered to be in the near-field region of the whole modular array, and at the same time, in the far-field region of each module. Thus, the phase variation brought by the quadratic phase term for each module in (\ref{eq:G_MRT_noation}) is rather small and $e^{j\kappa\frac{y_{n}^2}{2}\phi}$ could be neglected\cite{kosasih2026modular}. However, the phase accumulation of the cross term $e^{j\kappa Y_l y_n \phi}$ is non-negligible, as $Y_l$ can be large. Nevertheless, its effect is confined to the inner summation over $n$. We define $C_l(r, \theta) \triangleq \frac{1}{N} \sum_{n=0}^{N-1} e^{j\kappa y_n (\Delta + Y_l \phi)}$ and (\ref{eq:G_MRT_noation}) can be rewritten as 
\begin{equation}
    \label{eq:G_C_l}
    G(r, \theta) = \left| \frac{1}{L} \sum_{l=0}^{L-1} C_l(r, \theta) \cdot e^{j\kappa Y_l \Delta} e^{j\kappa \frac{Y_l^2}{2}\phi} \right|^2,
\end{equation}
 where the outer summation term $e^{j\kappa Y_l \Delta} e^{j\kappa \frac{Y_l^2}{2}\phi}$ varies with the angle difference $\Delta$ at a period of $\Delta_{L} = \lambda / d_s$, and the large inter-module spacing $d_s \gg \lambda$. In contrast, the variation of $C_l(r, \theta)$ with respect to $\Delta$ is dictated by $y_n$, whose maximum span is bounded by the aperture of each module, i.e., $Nd$. Since $Nd$ is generally several times smaller than $d_s$, the period of $C_l$ in the $\Delta$-domain is larger than $\Delta_L$. For example, for an M-ULA with $N=$\ 25 and $d_s=$\ 2.5 m at $f=$\ 15 GHz, $Nd$ is only one-tenth of $d_s$.  Thus, within the narrow angle interval between two adjacent grating lobes dictated by the outer summation term, $C_l$ varies slowly. 

Moreover, since $y_n$ is symmetrical about the module center, $C_l(r, \theta)$ is a real-valued \textit{sinc} function. Hence, $C_l(r, \theta)$ only modifies the amplitude of each term in the outer summation, without affecting the phase condition required for the outer summation to attain its peak. In other words, $C_l$ serves as a slow-varying spatial envelope that modulates the amplitude of the grating lobes pre-determined by
\begin{equation}
    \label{eq:G_L}
    G_L(r,\theta) \triangleq \frac{1}{L} \sum_{l=0}^{L-1} e^{j\kappa Y_{l}\Delta}e^{j\kappa \frac{Y_{l}^2}{2}\phi}.
\end{equation}

Therefore, the grating lobe analysis can be decoupled into two sequential steps. 
By analyzing the peak conditions of the outer summation $G_L$, we first locate the potential grating lobes, which are independent of the internal structure of each subarray. Then, we evaluate the residual power of these grating lobes via the envelope, to determine whether they fall outside the main lobe of the envelope and are effectively suppressed. Based on above, we will analyze near-field grating lobe locations for M-ULA in the next subsection.


\begin{figure}[!t]
  \centering 
  \subfigure[]{
  \includegraphics[width=0.48 \linewidth]{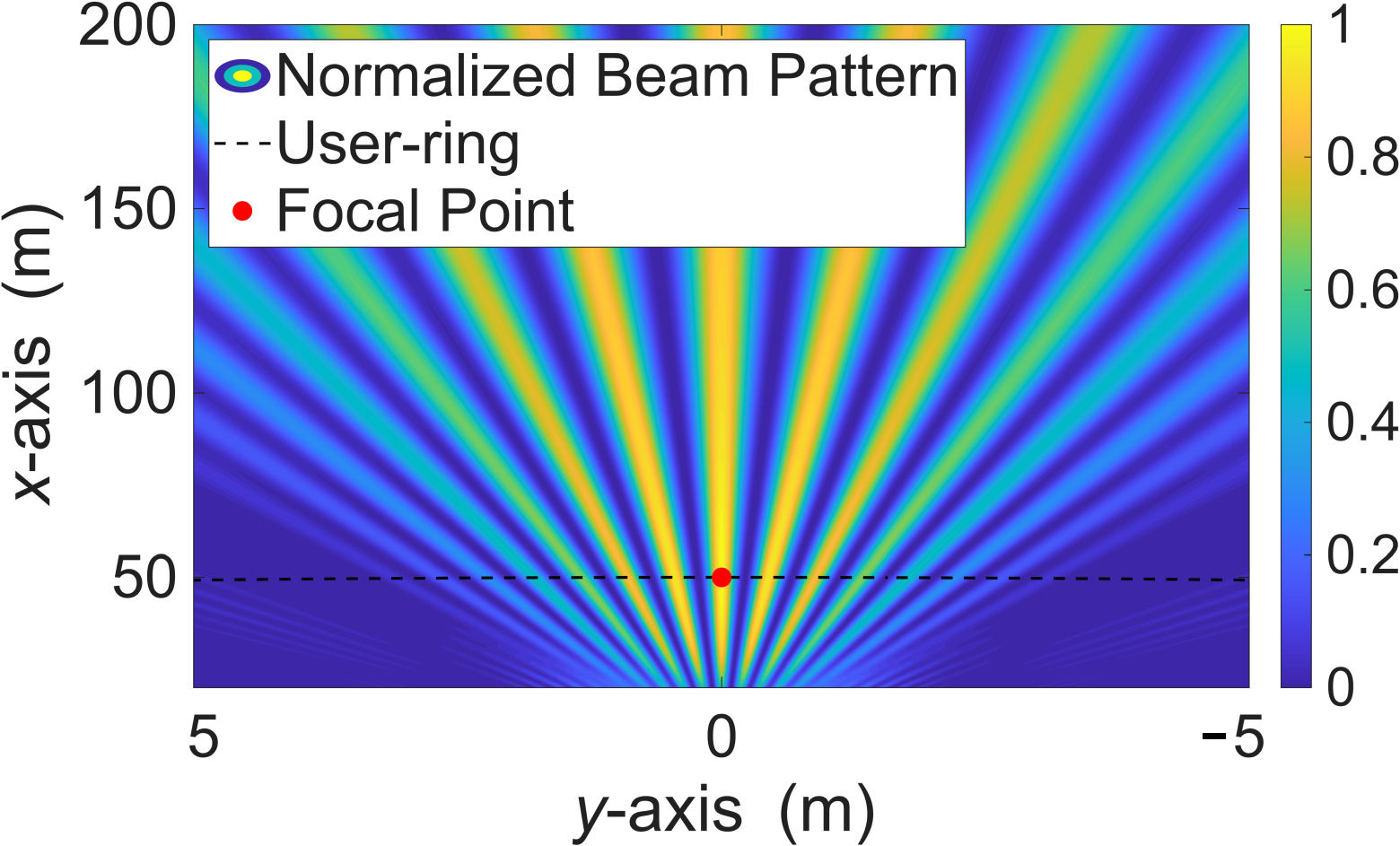}}
  \hfill
  \subfigure[]{
  \includegraphics[width=0.48 \linewidth]{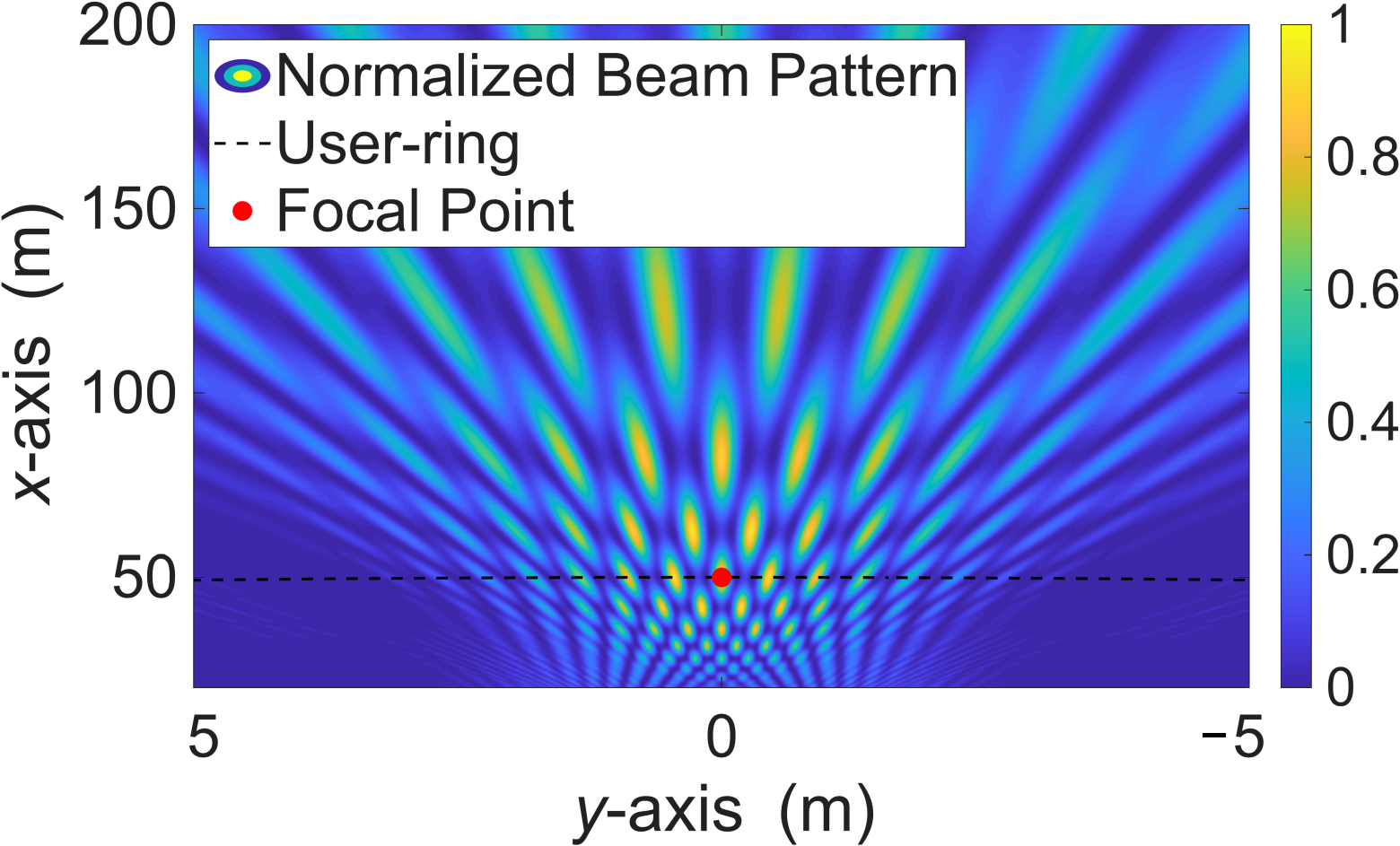}}
  \newline
  \subfigure[]{
  \includegraphics[width=0.48 \linewidth]{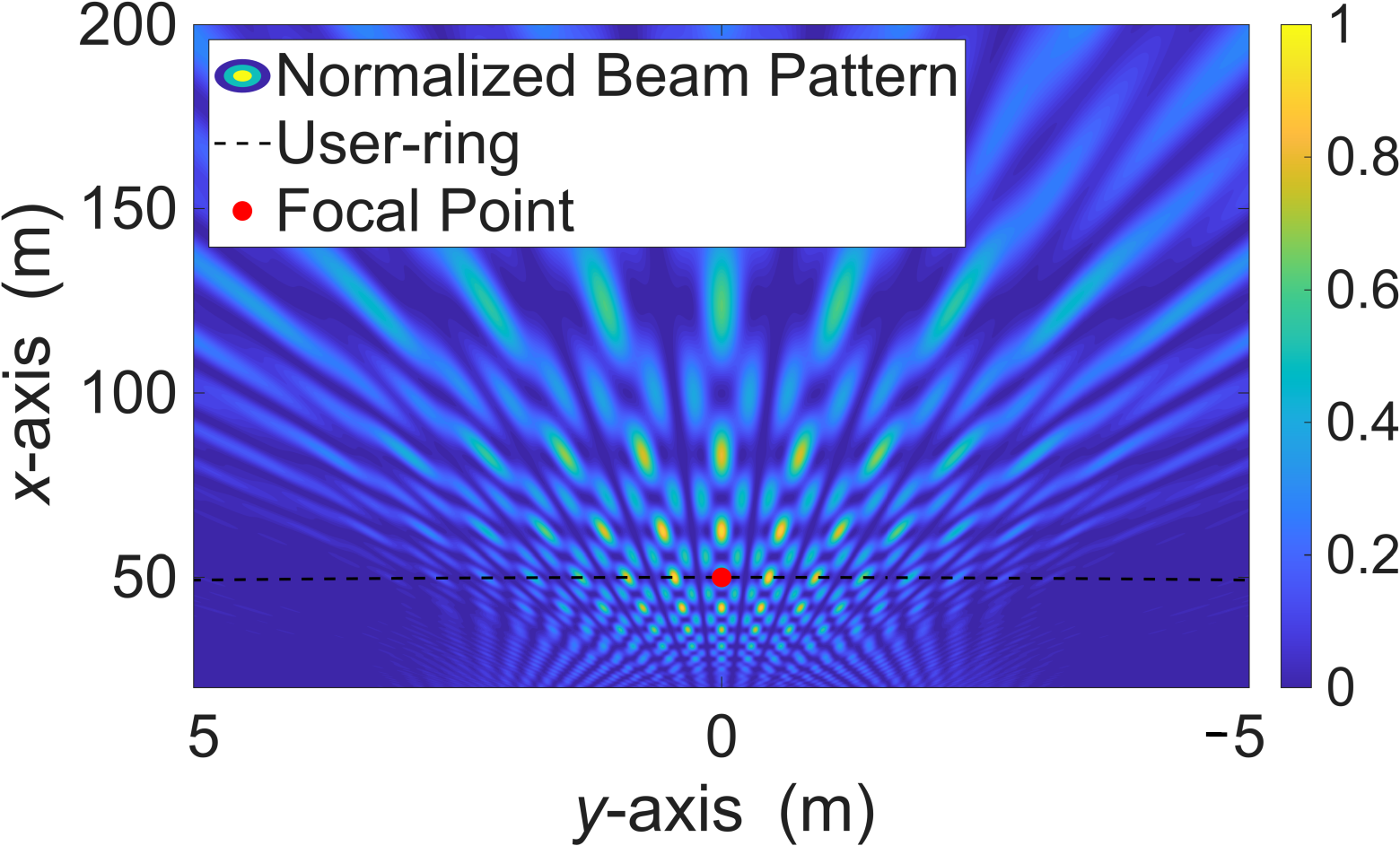}}
  \hfill
  \subfigure[]{
  \includegraphics[width=0.48 \linewidth]{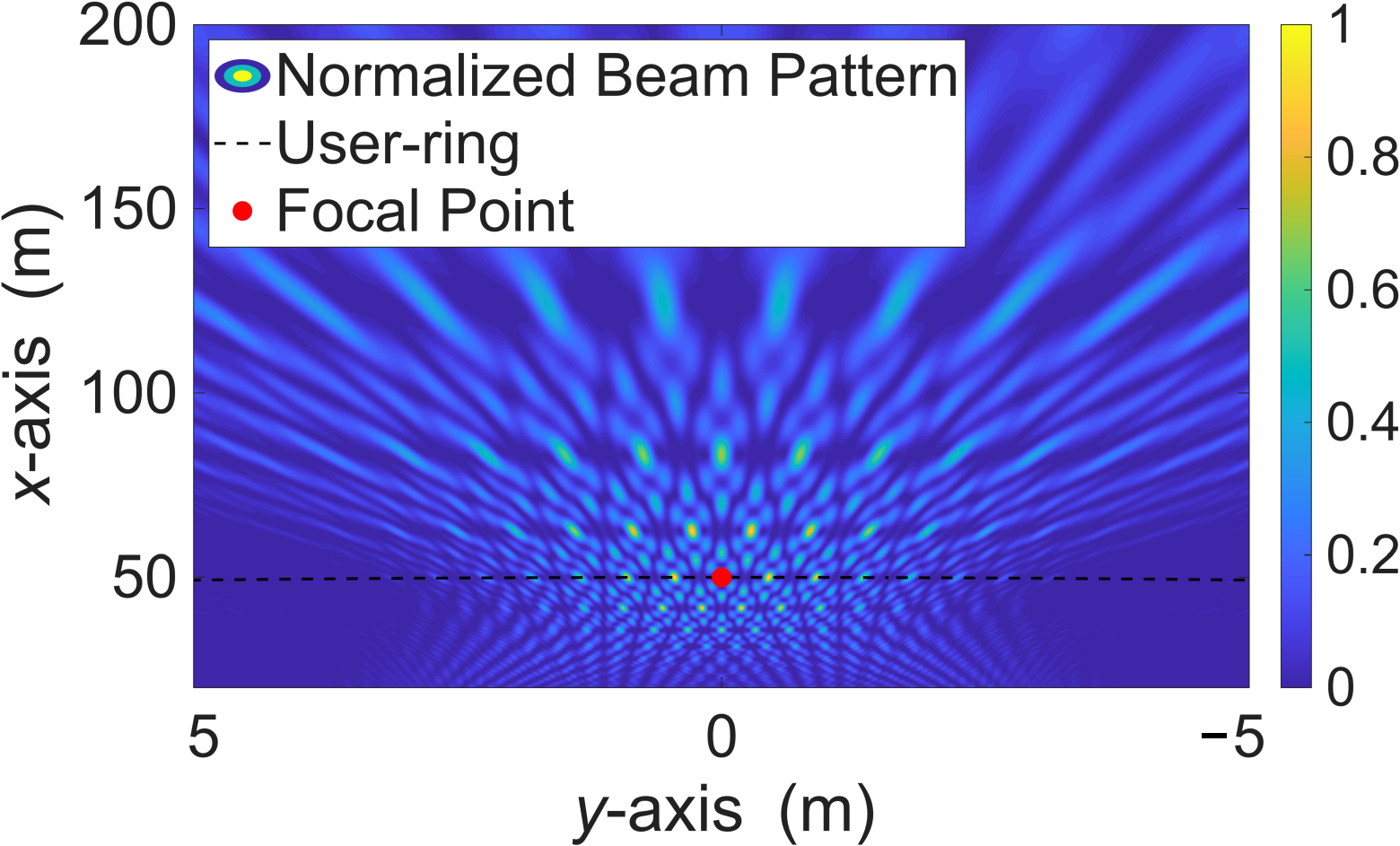}}
  \caption{Beam patterns of the M-ULA when $N=$\ 25, $d_{\mathrm{z}}=$\ 2 m, $r_{\mathrm{F}}=$\ 50 m, $\theta_{\mathrm{F}}=$\ 0 and $f=$\ 15 GHz: (a) $L=$\ 2; (b) $L=$\ 3; (c) $L=$\ 4; (d) $L=$\ 5.}
  \label{fig:beam_pattern}
\end{figure}

\subsection{Analysis of Near-Field Grating Lobe Location}
\label{location}

Unlike conventional far-field arrays, near-field grating lobe locations depend jointly on the angle-domain phase difference $\Delta$ and the distance-domain phase curvature $\phi$. Moreover, as shown in Fig.~\ref{fig:beam_pattern}, the beam pattern depends on whether the number of subarrays $L$ is even or odd, for the parity of $L$ leads to distinct symmetry of the subarray center distribution with respect to the array origin. When $L$ is even, the grating lobe arrangement in the near-field beam pattern exhibits aligned, otherwise interleaved.  

In this subsection, we will analyze the near-field grating lobe locations for these two cases, revealing that near-field grating lobes exist not only in the angle domain, but also in the distance domain, which is consistent with the physical resource in the distance domain within the near-field region. 

\begin{lemma}
For an M-ULA with \textbf{even} number of modules, the near-field grating lobe location can be expressed as 
\begin{align}
\label{eq:location_even}
\left\{ 
    \begin{aligned}
       & \sin\theta_n =\sin\theta_{\mathrm{F}}+n\frac{\lambda}{d_{\mathrm{s}}}, & n\in\mathbb{Z}, \\
       & r_m=\frac{r_{\mathrm{F}}d_{\mathrm{s}}^2\cos^2\theta}{d_{\mathrm{s}}^2\cos^2\theta_{\mathrm{F}} - m\lambda r_{\mathrm{F}}},   & m \in \mathbb{Z},\\
    \end{aligned}
\right.
\end{align}
where $\theta_{\mathrm{F}}$ and $r_{\mathrm{F}}$ is the focal angle and the focal distance, respectively. $\lambda$ is the carrier wavelength and $d_{\mathrm{s}}$ is the inter-module spacing.
\end{lemma}

\begin{IEEEproof}
We assume $L=2\ell,\ \ell \in \mathbb{Z}^+$. The subarray-center locations can be rewritten as $Y_l=\left(l+\frac{1}{2}\right)d_{\mathrm{s}}$, where $l=-\ell,\cdots,\ell-1$. Substituting $Y_l$ into (\ref{eq:G_L}) yields
\begin{equation}
    \label{eq:G_even}
    G_{L}^{\mathrm{even}}(r,\theta) = \frac{1}{L} \sum_{l=-\ell}^{\ell-1} e^{j\kappa d_\mathrm{s}\Delta (l+\frac{1}{2})} e^{j\frac{\kappa}{2}d_{\mathrm{s}}^2\phi (l^2 + l + \frac{1}{4})}.
\end{equation}

For convenience of subsequent analysis, we define
\begin{equation}
    \label{eq:beta}
    \beta \triangleq \frac{2\pi}{\lambda}d_\text{s}\Delta,
\end{equation}
\begin{equation}
    \label{eq:alpha}
    \alpha \triangleq \frac{\pi}{\lambda}d_\text{s}^2\phi,
\end{equation}
where $\beta$ denotes the linear phase coefficient, representing the linear phase due to angle differences between adjacent modules under the plane-wave approximation, and $\alpha$ denotes the quadratic phase coefficient, representing the quadratic phase accumulation under the spherical-wave model. In the near-field region, $\beta$ and $\alpha$ jointly determine the beam characteristics; in the far-field region, $\alpha$ tends to 0 as $\phi$ tends to 0, leaving only the $\beta$ term. Substituting $\beta$ and $\alpha$ into (\ref{eq:G_even}), we have
\begin{equation}
    \label{eq:G_even_notation}
    G_L^{\mathrm{even}}(r,\theta) = \frac{1}{L} e^{j(\frac{\beta}{2} + \frac{\alpha}{4})} \sum_{l=-\ell}^{\ell-1} e^{j\left(\beta l + \alpha(l^2 + l)\right)}.
\end{equation}
We note that $l(l+1)$ is always even when $L$ is even. Therefore, when $\alpha=\pi m\ (m \in \mathbb{Z})$, we have $e^{j\alpha l(l+1)}=1$ no matter what the exact value of $l$ is, and (\ref{eq:G_even_notation}) can be simplified to $G_L^{\mathrm{even}}(r,\theta) = \frac{1}{L} e^{j(\frac{\beta}{2} + \frac{\alpha}{4})} \sum_{l=-\ell}^{\ell-1} e^{j\beta l}$, which attains its peaks when $\beta=2\pi n\ (n \in \mathbb{Z})$. Based on the definition of $\Delta$, $\phi$, $\beta$ and $\alpha$, the grating lobes occur when
\begin{align}
\left\{ 
    \begin{aligned}
       & \frac{2\pi}{\lambda}d_{\mathrm{s}}\left(\sin\theta-\sin\theta_{\mathrm{F}}\right) = 2\pi n, & n\in\mathbb{Z}, \\
       & \frac{\pi d_{\mathrm{s}}^2}{\lambda}\left( \frac{\cos^2\theta_\textrm{F}}{r_\textrm{F}} - \frac{\cos^2\theta}{r} \right) = \pi m,   & m \in \mathbb{Z},\\
    \end{aligned}
\right.
\end{align}
from which we can derive (\ref{eq:location_even}), the grating lobe locations for M-ULA when $L$ is even. $\hfill \blacksquare$
\end{IEEEproof}

Contrary to the conventional belief, grating lobes exist not only in the angle domain, but also in the distance domain for modular arrays. The enlarged near-field range of modular arrays brings more physical resources in the distance domain corresponding to the quadratic phase term in the near-field model, endowing modular arrays with higher spatial resolution in the distance domain than conventional centralized arrays. From (\ref{eq:location_even}), we can also observe that, when $d_{\mathrm{s}}^2\cos^2\theta_{\mathrm{F}} < m\lambda r_{\mathrm{F}}$, the location of the $m-$th grating lobes in the distance domain is negative, thus the corresponding grating lobes do not exist. Therefore, the larger $d_{\mathrm{s}}$ is, or the closer to zero $\theta_{\mathrm{F}}$ is, the more values of $m$ satisfying $d_{\mathrm{s}}^2\cos^2\theta_{\mathrm{F}} - m\lambda r_{\mathrm{F}} > 0$ there exist, and the more grating lobes in the distance domain there are, as shown in Fig.~\ref{fig:lobe_numbers_ds}. 

\begin{figure}[t!]
  \centering 
  \subfigure[]{
  \includegraphics[width=0.47 \linewidth]{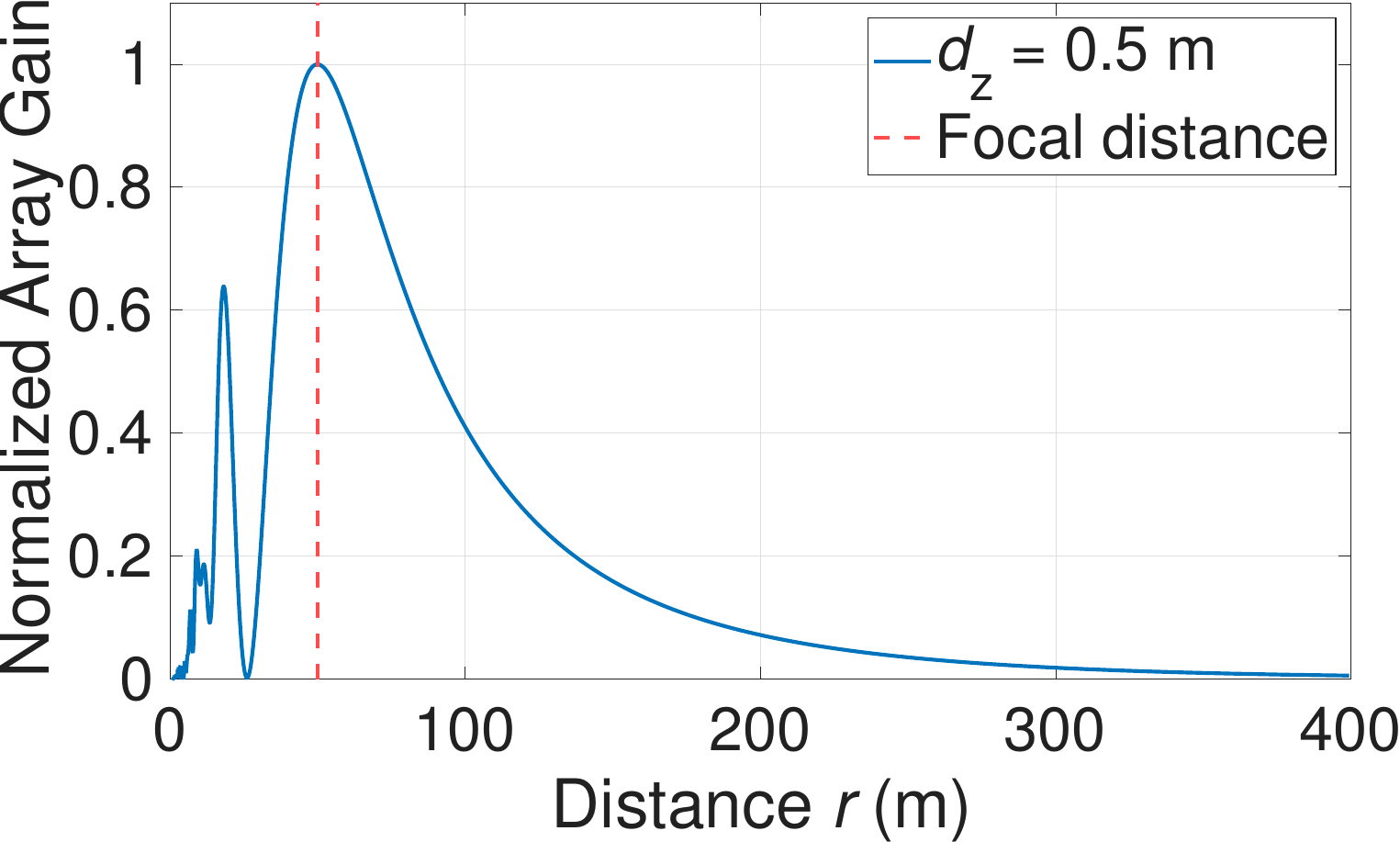}}
  \hfill
  \subfigure[]{
  \includegraphics[width=0.47 \linewidth]{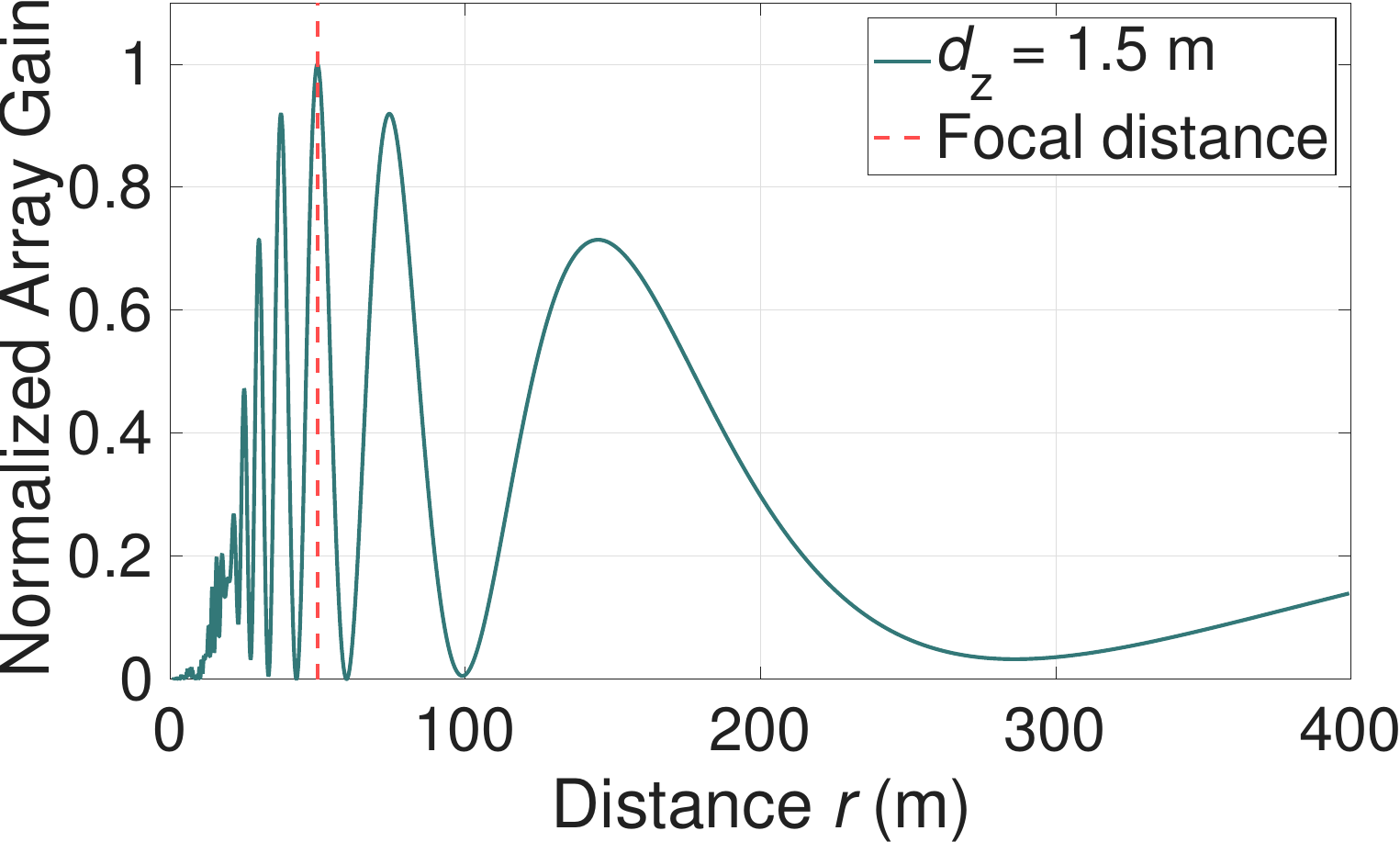}}
  \caption{Normalized array gain along $\theta=$\ 0 when $N=$\ 25, $L=$\ 4, $r_{\mathrm{F}}=$\ 50 m, $\theta_{\mathrm{F}}=$\ 0 and $f=$\ 15 GHz: (a) $d_{\mathrm{z}}=$\ 0.5 m; (b) $d_{\mathrm{z}}=$\ 1.5 m.}
  \label{fig:lobe_numbers_ds}
\end{figure}

Specifically, $L=$\ 2 is a special case of the analysis discussed above. 
 For a fixed observation angle, the variation of array gain with distance comes mainly from the response of the antennas within each subarray. Since the subarray aperture is rather small, the resolution in the distance domain is limited, making the grating lobes in the distance domain seldom observed when $L=$\ 2, as shown in Fig.~\ref{fig:beam_pattern}\ (a). The corresponding mathematical explanation is given in Appendix \ref{distinctiveness_L_2}.
 
 To summarize, when $L$ is an even number, the subarray centers are symmetrically paired on both sides of the origin, with no subarray located exactly at the center. Since the symmetric pairs always appear together, the quadratic phase variations induced by spherical wavefront curvature remain synchronized between the left and right sides, preserving the periodic coherence condition in the angle domain. Consequently, grating lobes in the distance domain occur at integer multiples of the angle period, forming a regularly aligned beam pattern, which is different from the beam pattern when $L$ is an odd number.

\begin{lemma}
For an M-ULA with \textbf{odd} number of modules, the near-field grating lobe location can be expressed as
\begin{align}
\label{eq:location_odd_1}
\left\{ 
    \begin{aligned}
       & \sin\theta_q =\sin\theta_{\mathrm{F}}+q\frac{\lambda}{d_{\mathrm{s}}}, & q\in\mathbb{Z}, \\
       & r_{2p}=\frac{r_{\mathrm{F}}d_{\mathrm{s}}^2\cos^2\theta}{d_{\mathrm{s}}^2\cos^2\theta_{\mathrm{F}} - 2p\lambda r_{\mathrm{F}}},   & p \in \mathbb{Z},\\
    \end{aligned}
\right.
\end{align}

\begin{align}
\label{eq:location_odd_2}
\left\{ 
    \begin{aligned}
       & \sin\theta_q =\sin\theta_{\mathrm{F}} + \frac{2q+1}{2} \frac{\lambda}{d_{\mathrm{s}}}, & q\in\mathbb{Z}, \\
       & r_{2p+1}=\frac{r_{\mathrm{F}}d_{\mathrm{s}}^2\cos^2\theta}{d_{\mathrm{s}}^2\cos^2\theta_{\mathrm{F}} - (2p+1)\lambda r_{\mathrm{F}}},   & p \in \mathbb{Z}.\\
    \end{aligned}
\right.
\end{align}
\end{lemma}

\begin{IEEEproof}
We assume $L=2\ell+1,\ \ell \in \mathbb{N}$. Then the subarray-center locations can be rewritten as $Y_l=ld_{\mathrm{s}}$, where $l=-\ell,\cdots,\ell$. Under the definition of $\beta$ and $\alpha$ in (\ref{eq:beta}) and (\ref{eq:alpha}), (\ref{eq:G_L}) can be expressed as
 \begin{equation}
     \label{eq:G_odd_notation}
     G_{L}^{\mathrm{odd}}(r,\theta) = \frac{1}{L} \sum_{l=-\ell}^{\ell} e^{j\left(\beta l+\alpha l^2 \right)}.
 \end{equation}
 We note that $l^2$ and $l$ are either both odd or both even, i.e.,$l^2 \equiv l \pmod{2}$. As the integer factor for $\alpha$, $l^2$ is no longer always an even number. Therefore, the condition that $\beta$ and $\alpha$ need to satisfy to form a grating lobe changes as $l$ varies.
 
 When $\alpha = 2\pi p\ (p\in\mathbb{Z})$, we always have $e^{j\alpha l^2}=1$, and (\ref{eq:G_odd_notation}) reduces to a standard Dirichlet kernel:
  \begin{equation}
     G_{L}^{\mathrm{odd}}(r,\theta) = \frac{1}{L} \sum_{l=-\ell}^{\ell} e^{j\beta l},
 \end{equation}
 which attains its peak when $\beta = 2\pi q\ (q\in\mathbb{Z})$. Similar to the analysis when $L$ is even, we can obtain (\ref{eq:location_odd_1}).

When $\alpha = (2p+1)\pi\ (p\in\mathbb{Z})$, we have $e^{j\alpha l^2} = (-1)^{l^2} = (-1)^{l} = e^{j\pi l}$, and (\ref{eq:G_odd_notation}) can be simplified as
\begin{equation}
    G_{L}^{\mathrm{odd}}(r,\theta) = \frac{1}{L} \sum_{l=-\ell}^{\ell} e^{j\left(\beta+\pi \right) l},
\end{equation}
which attains its peak when $\beta + \pi = 2\pi q$, i.e., $\beta = (2q+1)\pi$, $q\in\mathbb{Z}$. In this case, the corresponding locations of grating lobes are (\ref{eq:location_odd_2}). $\hfill \blacksquare$
\end{IEEEproof}


Through above analysis, the reason why the near-field beam patterns exhibit interleaved grating lobe arrangements when $L$ is odd can be explained. When $L$ is odd, in addition to the symmetric pairs, one subarray is situated precisely at the array center, acting as a fixed phase reference as the observation distance varies. At certain values of the spherical wavefront curvature, an extra $\pi$ phase difference arises between the symmetric pairs and the center subarray, making the odd-order distance grating lobes fall at half-integer angle locations.

Particularly for $L=$\ 1, there are always $\ell=0$ and $l=0$, making each summation term in (\ref{eq:G_odd_notation}) equal to 1. In this case, there are no inter-subarray interferences, and no high-gain grating lobes will be generated due to coherent interferences. Therefore, near-field grating lobes in the distance domain are rarely observed in conventional centralized arrays.

As can be seen from (\ref{eq:location_even}), (\ref{eq:location_odd_1}), and (\ref{eq:location_odd_2}), once $\lambda$, $d_{\mathrm{s}}$, $r_{\mathrm{F}}$ and $\theta_{\mathrm{F}}$ are determined, the inherent grating lobe locations are independent of the specific value of $L$. 
The specific value of $L$ only affects the array gain peak and the beamfocusing spatial extent, which is shown in Fig.~\ref{fig:beam_pattern}\ (b) and Fig.~\ref{fig:beam_pattern}\ (d). 
As $L$ increases, $|Y_l|$, the distance from the outer subarray to the array center, increases, leading to a more severe attenuation of the amplitude envelope $C_l$ in the non-focusing direction. Since the total array gain is the product of the in-phase term between subarrays and the amplitude envelope $C_l$, the array gain peak of the grating lobes is suppressed. Therefore, the larger $L$ is, the fewer near-field grating lobes with higher array gain peaks there are, and the more the near-field grating lobes are suppressed.

In this subsection, we reveal that the near-field grating lobes of modular arrays not only exist in the angular domain, but also in the distance domain. The reason why the parity of $L$ affects the arrangement of the beam pattern is also explained. 

\subsection{Analysis of Near-Field Beamfocusing Spatial Extent}

To better characterize the beam pattern of an M-ULA, we analyze the near-field beamfocusing spatial extent in this subsection, using the classic definition of half-power beamwidth (BW) and beamdepth (BD).

\subsubsection{\textbf{Analysis of Beamwidth}}
\label{BW}

For the main focusing region, the 3-dB beamwidth refers to the half-power spatial angle along the user-ring, where there are a set of angle-distance pairs with a nearly constant beamwidth. Suppose that the spatial angles for two half-power locations are $\theta_{\textrm{small}}$ and $\theta_{\textrm{large}}$ respectively, then $\mathrm{BW}=\left| \sin\theta_{\textrm{large}} - \sin\theta_{\textrm{small}} \right|$. Under the classic definition of beamwidth mentioned above, we consider the user-ring condition $\phi=0$, where the quadratic phase term in (\ref{eq:G_MRT}) vanishes, and (\ref{eq:G_MRT}) could be decoupled into two Dirichlet kernels only related to the angle difference $\Delta$:
\begin{equation}
  \label{eq:G_BW}
  G(\theta;\theta_\text{F}) = \left|G_{N}(\theta)\right|^2 \cdot \left|G_{L}(\theta)\right|^2,
\end{equation}
where $G_N(\theta) = \frac{1}{N}\sum_{n=0}^{N-1} e^{j\kappa y_n\Delta}$ is the Dirichlet kernel contributed by each module, and $G_L(\theta) = \frac{1}{L}\sum_{l=0}^{L-1} e^{j\kappa Y_l\Delta}$ is the Dirichlet kernel contributed by the inter-module structure. When $\theta$ is close to $\theta_{\mathrm{F}}$, we substitute $y_{n} = (n-\frac{N-1}{2})d$ and $Y_{l} = (l-\frac{L-1}{2})d_{\mathrm{s}}$ into (\ref{eq:G_MRT}), and obtain
\begin{equation}
  \label{eq:G_BW_approx}
  G(\theta;\theta_\text{F}) \approx \left| \operatorname{sinc}\left( \frac{N\Delta}{2} \right) \right|^2 \cdot \left| \operatorname{sinc}\left( \frac{Ld_{\mathrm{s}}\Delta}{\lambda} \right) \right|^2,
\end{equation}
where $\operatorname{sinc}(x) = \frac{\sin(\pi x)}{\pi x}$. 
Since the aperture of a single module typically differs from that of the overall array by more than a factor of ten, $ \left|G_{N}(\theta)\right|^2$ varies rather slowly and acts as a wide envelope. Within the narrow angular region of the main lobe, $ \left|G_{N}(\theta)\right|^2$ remains approximately constant. Consequently, the 3-dB condition could be simplified as $\left|G_{L}(\theta_\mathrm{3dB})\right|^2 =$\ 0.5. Solving $\left| \operatorname{sinc}\left( \frac{Ld_{\mathrm{s}}\Delta}{\lambda} \right) \right|^2=$\ 0.5 yields $\frac{Ld_{\mathrm{s}}\Delta}{\lambda}=$\ 0.443. Assuming $\theta_{\mathrm{F}}=$\ 0, the half-power beamwidth for main focusing region is given by
\begin{equation}
  \label{eq:BW}
  \mathrm{BW} = \frac{0.886\lambda}{Ld_{\mathrm{s}}}.
\end{equation}

\subsubsection{\textbf{Analysis of Beamdepth}}
\label{BD}


\begin{figure}[t!]
	\centering
	\includegraphics[width=0.38\textwidth]{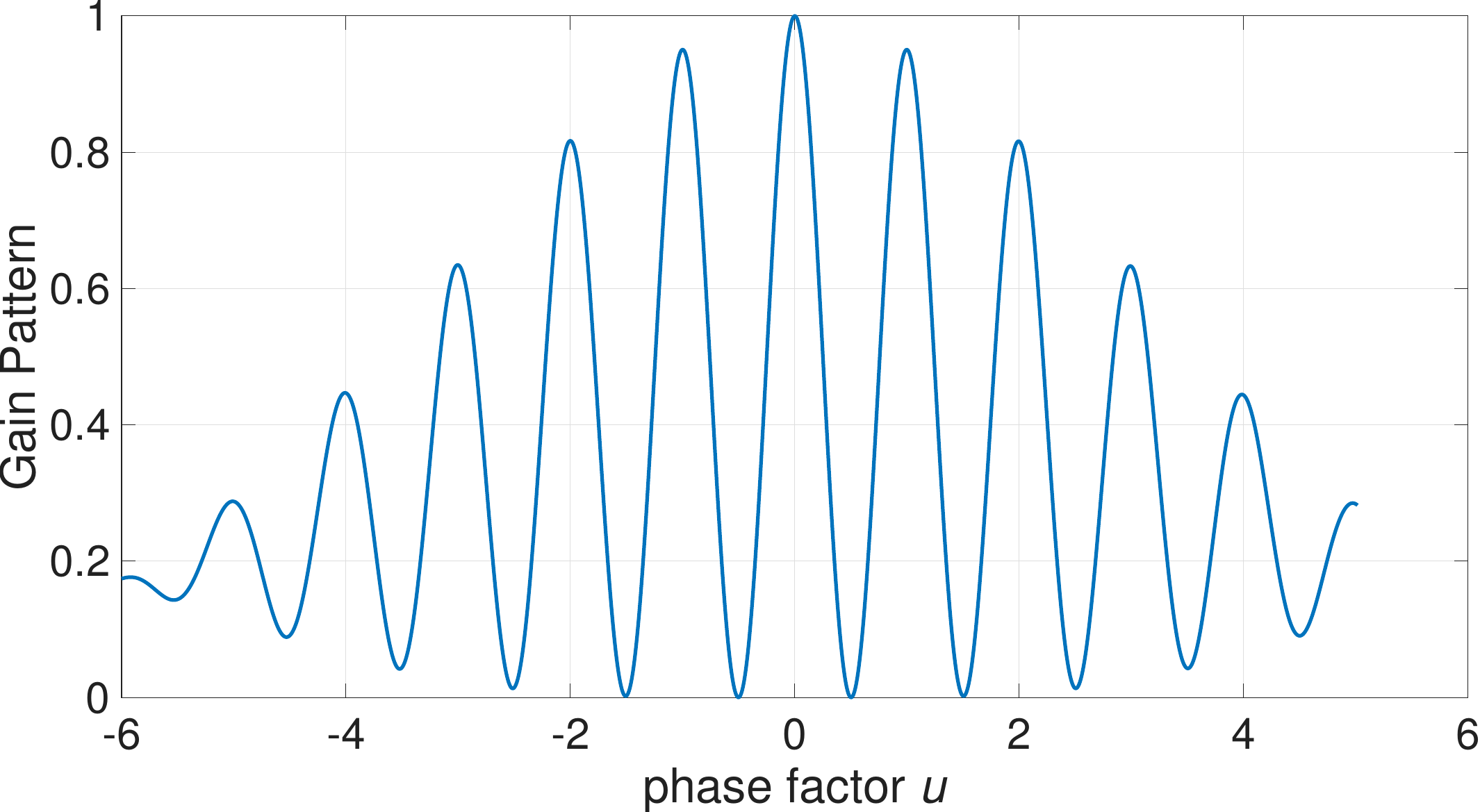}
	\caption{Gain pattern along $\theta=$\ 0 for $L=$\ 4, $N=$\ 25 and $d_{\mathrm{z}}=$\ 2 m.}
	\label{fig:G_u}
\end{figure}

\begin{figure}[t!]
  \centering 
  \subfigure[]{
  \includegraphics[width=0.47 \linewidth]{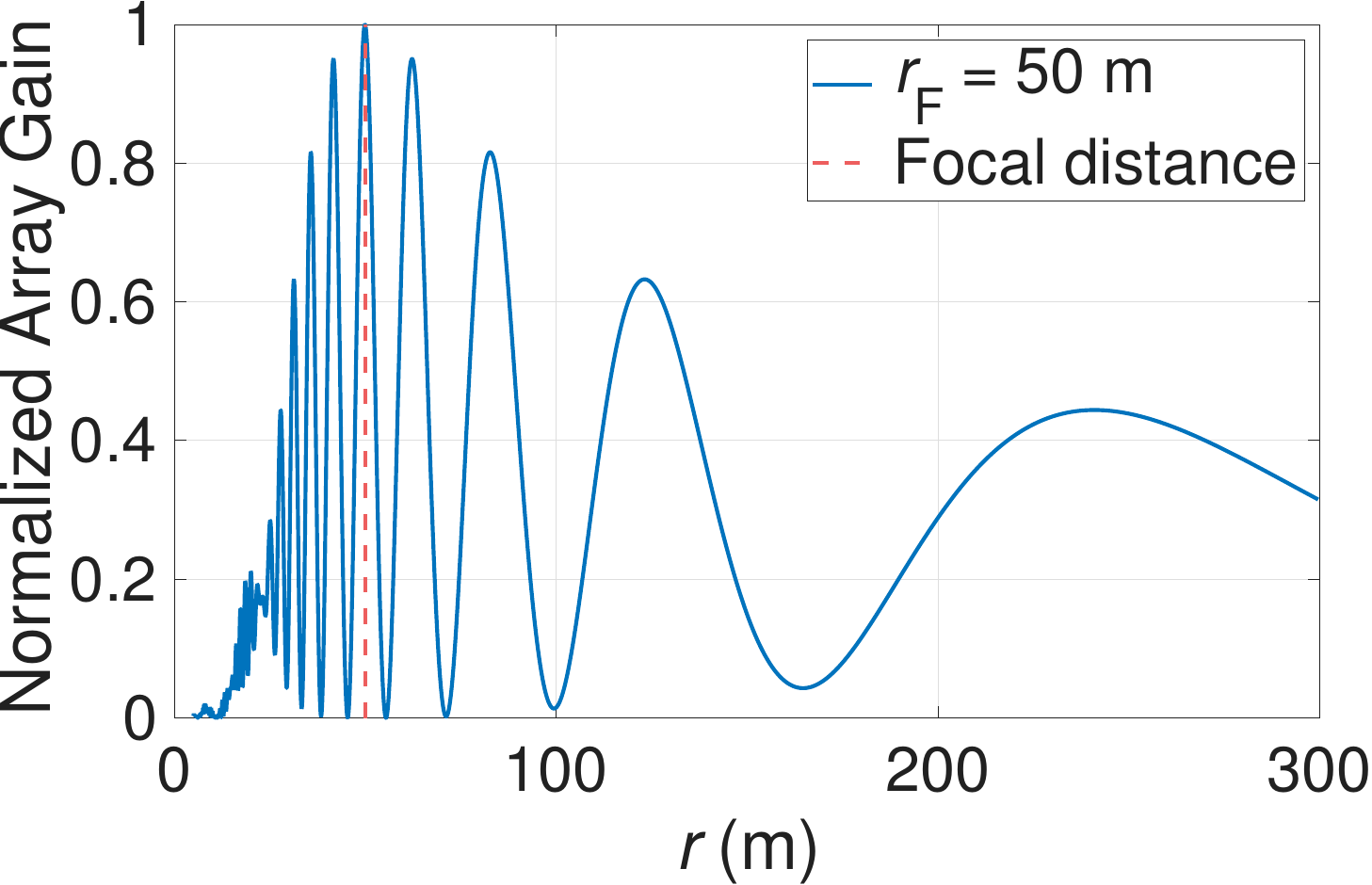}}
  \hfill
  \subfigure[]{
  \includegraphics[width=0.47 \linewidth]{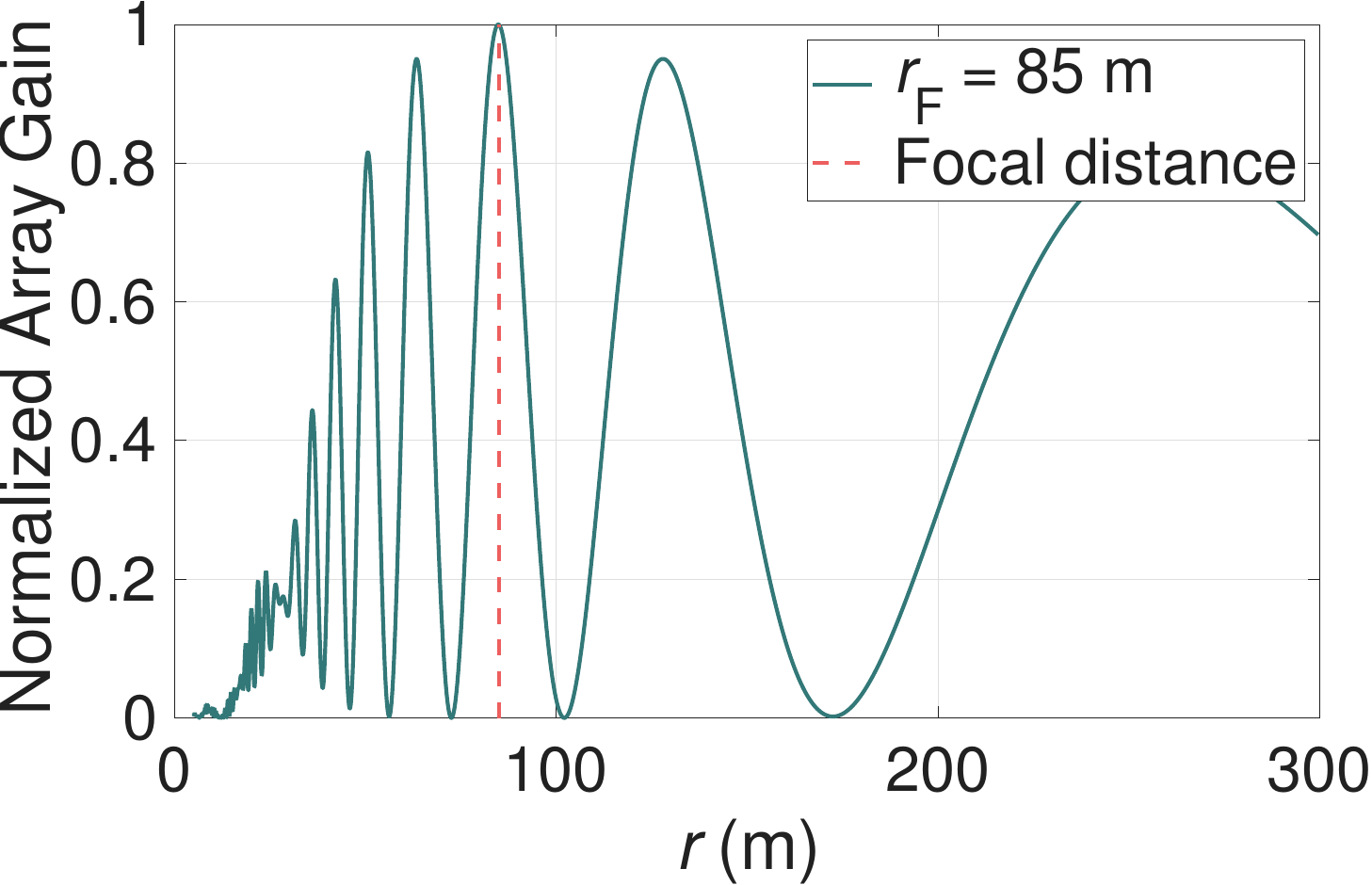}}
  \caption{Normalized array gain along $\theta=$\ 0 under the same array structure with different $r_{\mathrm{F}}$: (a) $r_{\mathrm{F}}=$\ 50 m; (b) $r_{\mathrm{F}}=$\ 85 m.}
  \label{fig:gain_pattern}
\end{figure}

For the main focusing region, the 3-dB beamdepth refers to the half-power spatial distance along the $\theta = \theta_{\mathrm{F}}$ line in the distance domain. We assume that the spatial distances for two half-power locations are $r_{\textrm{small}}$ and $r_{\textrm{large}}$ respectively, then $\mathrm{BD}=\left| r_{\textrm{large}} - r_{\textrm{small}} \right|$. Different from the beamwidth, even for the same array configuration, the beamdepth varies with the focal distance $r_{\mathrm{F}}$. To enable a unified analysis of beamdepth for different $r_{\mathrm{F}}$, we introduce the concept of \textit{gain pattern} $G(u)$, treating the normalized array gain in (\ref{eq:G_MRT}) as a function of the dimensionless quantity
\begin{equation}
  \label{eq:u}
  u \triangleq \phi \frac{d_{\mathrm{s}}^2}{\lambda} = \frac{d_{\mathrm{s}}^2}{\lambda}\left( \frac{\cos^2\theta_\textrm{F}}{r_\textrm{F}} - \frac{\cos^2\theta}{r} \right),
\end{equation}
which represents the quadratic phase factor. Under the definition in (\ref{eq:alpha}), we have $u = \frac{\alpha}{\pi}$. Based on the analysis above, integer values of $u$ correspond to the grating lobe locations, so that the same array structure with different focal distances $r_{\mathrm{F}}$ share the same gain pattern $G(u)$. In other words, as $r_{\mathrm{F}}$ changes, grating lobes of the same order attain the same peak gain, with only their exact locations varying. 


Fig.~\ref{fig:G_u} shows the gain pattern for a fixed array structure, and Fig.~\ref{fig:gain_pattern} illustrates the impact of $r_{\mathrm{F}}$ on the grating lobe location and beamdepth. As can be seen, the gain pattern is determined by the array configuration, and different values of $r_{\mathrm{F}}$ only cause stretching transformations to $G(u)$ to varying degrees. Therefore, under a certain array structure, as long as the values of $u_\mathrm{3dB}^\mathrm{small}$ and $u_\mathrm{3dB}^\mathrm{large}$ are determined, we can obtain the corresponding $r_\mathrm{3dB}^\mathrm{small}$ and $r_\mathrm{3dB}^\mathrm{large}$ according to (\ref{eq:u}):
\begin{equation}
    \label{eq:r_3dB}
    r_\mathrm{3dB}^\mathrm{large/small}=\frac{r_{\mathrm{F}}d_{\mathrm{s}}^2\cos^2\theta}{d_{\mathrm{s}}^2\cos^2\theta_{\mathrm{F}} - u_\mathrm{3dB}^\mathrm{large/small}\lambda r_{\mathrm{F}}},
\end{equation}
based on which we can calculate the beamdepth for each high-gain region as (\ref{eq:BD}). For the main focusing region, we have $u_\mathrm{3dB}^\mathrm{small} = -u_\mathrm{3dB}^\mathrm{large}$, so that (\ref{eq:BD}) could be simplified as
\begin{figure*}[!t]
    \begin{equation}
    \label{eq:BD}
    \mathrm{BD} = r_\mathrm{3dB}^\mathrm{large} - r_\mathrm{3dB}^\mathrm{small}
    =\frac{\lambda r_{\mathrm{F}}^2 d_{\mathrm{s}}^2\cos^2\theta \left(u_\mathrm{3dB}^\mathrm{large} - u_\mathrm{3dB}^\mathrm{small}\right)}{d_{\mathrm{s}}^4\cos^4\theta_{\mathrm{F}} - \lambda r_{\mathrm{F}}d_{\mathrm{s}}^2\cos^2\theta_{\mathrm{F}}\left( u_\mathrm{3dB}^\mathrm{large} + u_\mathrm{3dB}^\mathrm{small} \right) + \lambda^2 r_{\mathrm{F}}^2 u_\mathrm{3dB}^\mathrm{large} u_\mathrm{3dB}^\mathrm{small} }. 
    \end{equation}
    \hrulefill
\end{figure*}
\begin{equation}
    \label{eq:BD_main}
    \mathrm{BD}_\mathrm{main} =\frac{2\lambda r_{\mathrm{F}}^2 d_{\mathrm{s}}^2\cos^2\theta \ u_\mathrm{3dB}^\mathrm{large}} {d_{\mathrm{s}}^4\cos^4\theta_{\mathrm{F}}  - \lambda^2 r_{\mathrm{F}}^2 (u_\mathrm{3dB}^\mathrm{large})^2 }.
\end{equation}
Notably, for the same array structure with different focal distances $r_{\mathrm{F,a}}$ and $r_{\mathrm{F,b}}$, the ratio of the beamdepths of their main focusing regions could be expressed as 
\begin{equation}
  \label{eq:BD_ratio}
  \frac{\mathrm{BD}_\mathrm{main,a}}{\mathrm{BD}_\mathrm{main,b}} = \left(\frac{r_{\mathrm{F,a}}}{r_{\mathrm{F,b}}}\right)^2 \frac{d_{\mathrm{s}}^4\cos^4\theta_{\mathrm{F}}  - \lambda^2 r_{\mathrm{F,b}}^2 (u_\mathrm{3dB}^\mathrm{large})^2}{d_{\mathrm{s}}^4\cos^4\theta_{\mathrm{F}}  - \lambda^2 r_{\mathrm{F,a}}^2 (u_\mathrm{3dB}^\mathrm{large})^2}.
\end{equation}
Due to the enlarged module spacing, modular arrays are likely to satisfy $\left(d_{\mathrm{s}}\cos\theta_\mathrm{F}\right)^2 \gg \lambda r_{\mathrm{F}}u_\mathrm{3dB}^\mathrm{large}$. In this case, the latter item could be approximated as \text{1} and we get
\begin{equation}
  \label{eq:BD_ratio_approx}
  \frac{\mathrm{BD}_\mathrm{main,a}}{\mathrm{BD}_\mathrm{main,b}} \approx \left(\frac{r_{\mathrm{F,a}}}{r_{\mathrm{F,b}}}\right)^2,
\end{equation}
from which we can see that, the beamdepth of the main focusing region is approximately proportional to $r_{\mathrm{F}}^2$. That's exactly why the mapping of the gain pattern $G(u)$ to the distance domain broadens significantly as the focal distance $r_{\mathrm{F}}$ increases. Analogous to the concept of \textit{depth of field} for optical lenses, the antenna array acts as a phased-array lens, focusing in the distance dimension by controlling the curvature of spherical waves. The larger the array aperture is, the farther the focal distance is enabled, making the focal depth longer, i.e., the beamdepth larger.

Besides, when $\left(d_{\mathrm{s}}\cos\theta_\mathrm{F}\right)^2 \gg \lambda r_{\mathrm{F}}u_\mathrm{3dB}^\mathrm{large}$ is satisfied, we can also observe from (\ref{eq:BD_main}) that when $r_{\mathrm{F}}$ remains the same, the beamdepth of the main focusing region is inversely proportional to $d_{\mathrm{s}}^2$, for increasing the module spacing will enlarge the effective array aperture, making the path differences between modules more sensitive to distance changes, thereby improving the distance resolution.  Combined with the discussion in Section~\ref{location}, we have found that enlarging the module spacing $d_\mathrm{s}$ will, on one hand, increase the number of observable near-field grating lobes, generating more interference regions with high array gain, and on the other hand, reduce the beamdepth of each grating lobe, enhancing the accuracy of near-field beamfocusing. Interestingly, the benefits and drawbacks coexist as $d_{\mathrm{s}}$ increases, so that there is likely to be an inflection point regarding $d_{\mathrm{s}}$ where the effects of the two can be optimally balanced, as shown in Section~\ref{sim}.


In this section, we have analyzed the beam pattern for M-ULA. we will further study how to suppress near-field grating lobes for modular arrays in the following section.

\section{Near-Field Grating Lobe Suppression for Modular Array}



In this section, we first analyze how to suppress near-field grating lobes for M-ULA. The results demonstrate that, increasing the number of antennas per
module works, which also means that the grating lobe interferences are severe under a limited number of antennas. To break this limitation, we further propose the modular minimum-redundancy array (M-MRA) by redesigning the subarray configuration, which can suppress near-field grating lobes in both the angle and distance domains simultaneously with the same number of antennas.

\subsection{Single-region Beamfocusing for M-ULA}

In Section~\ref{section_analysis}, we have explained that for M-ULA, the slowly-varying envelope mainly depends on the structure of each single module, while the array factor determined by the arrangement of all modules varies rapidly. Naturally, if the first grating lobes of the rapidly-varying term are located outside the main lobe of the envelope, these grating lobes will be suppressed by the envelope, preserving only a main focusing region in the overall array gain pattern. We will analyze this single-region beamfocusing condition in this subsection.

\begin{figure}[!b]
  \centering 
  \subfigure[]{
  \includegraphics[width=0.48 \linewidth]{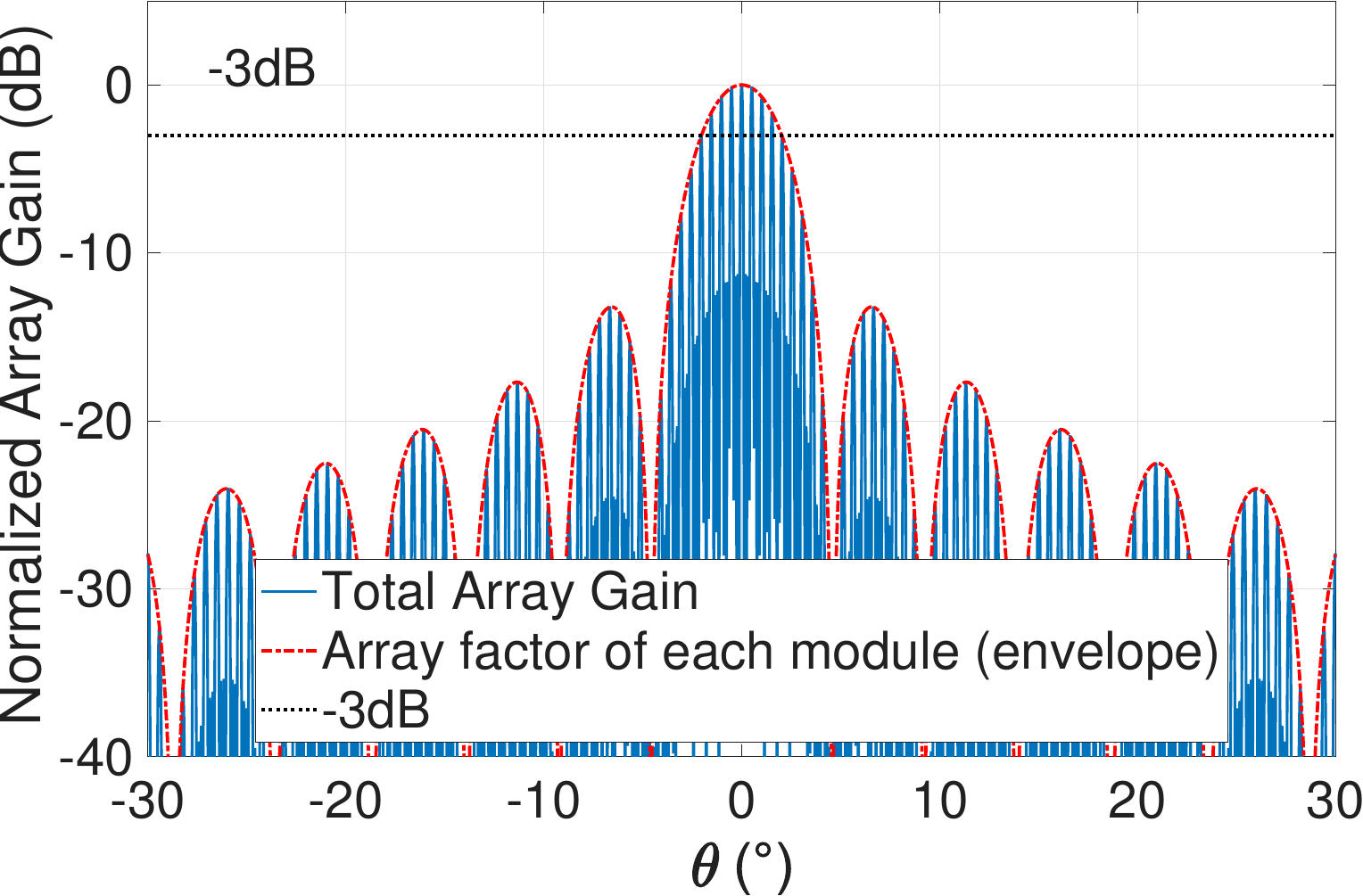}}
  \hfill
  \subfigure[]{
  \includegraphics[width=0.48 \linewidth]{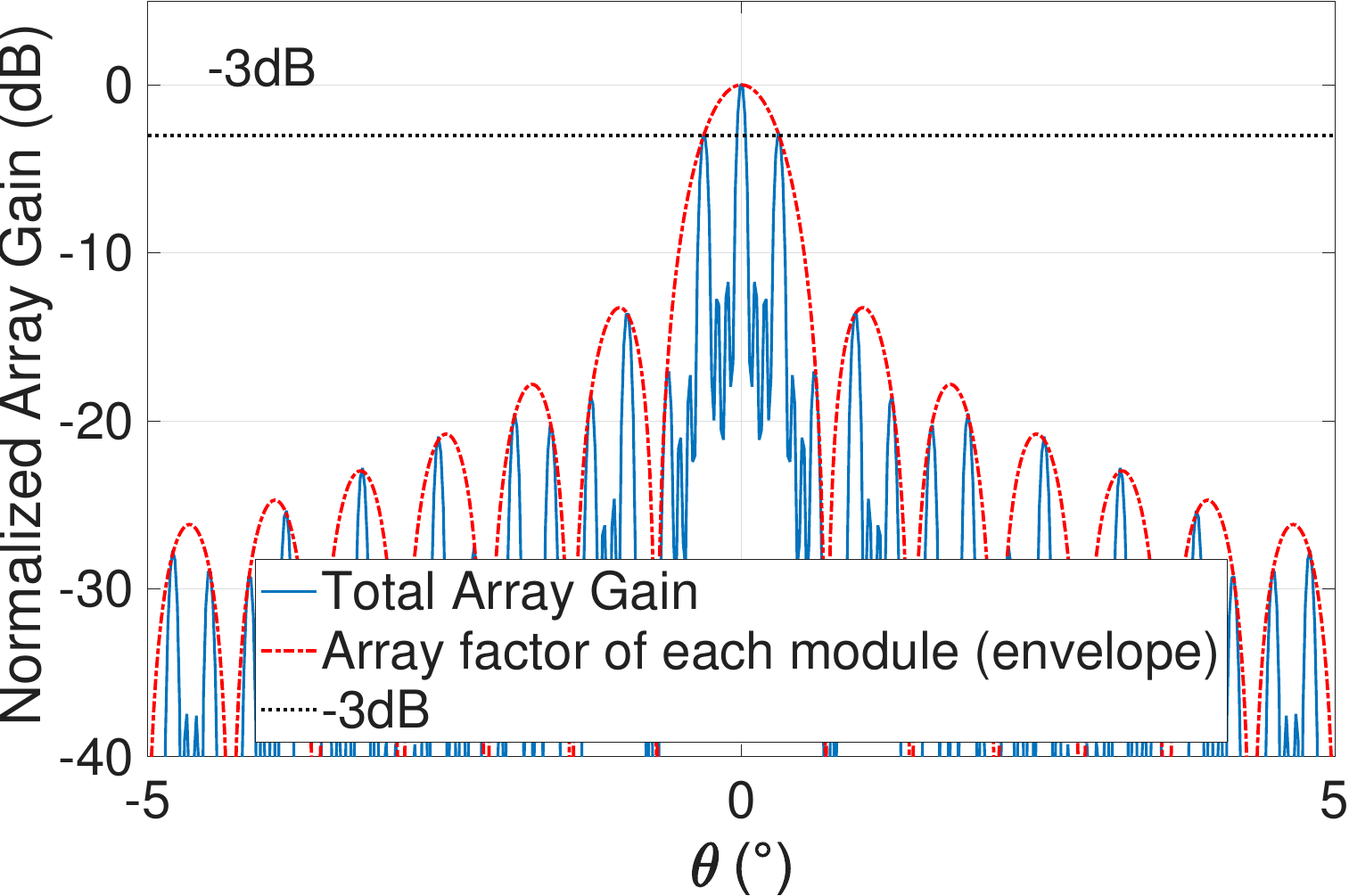}}
  \newline
  \subfigure[]{
  \includegraphics[width=0.475 \linewidth]{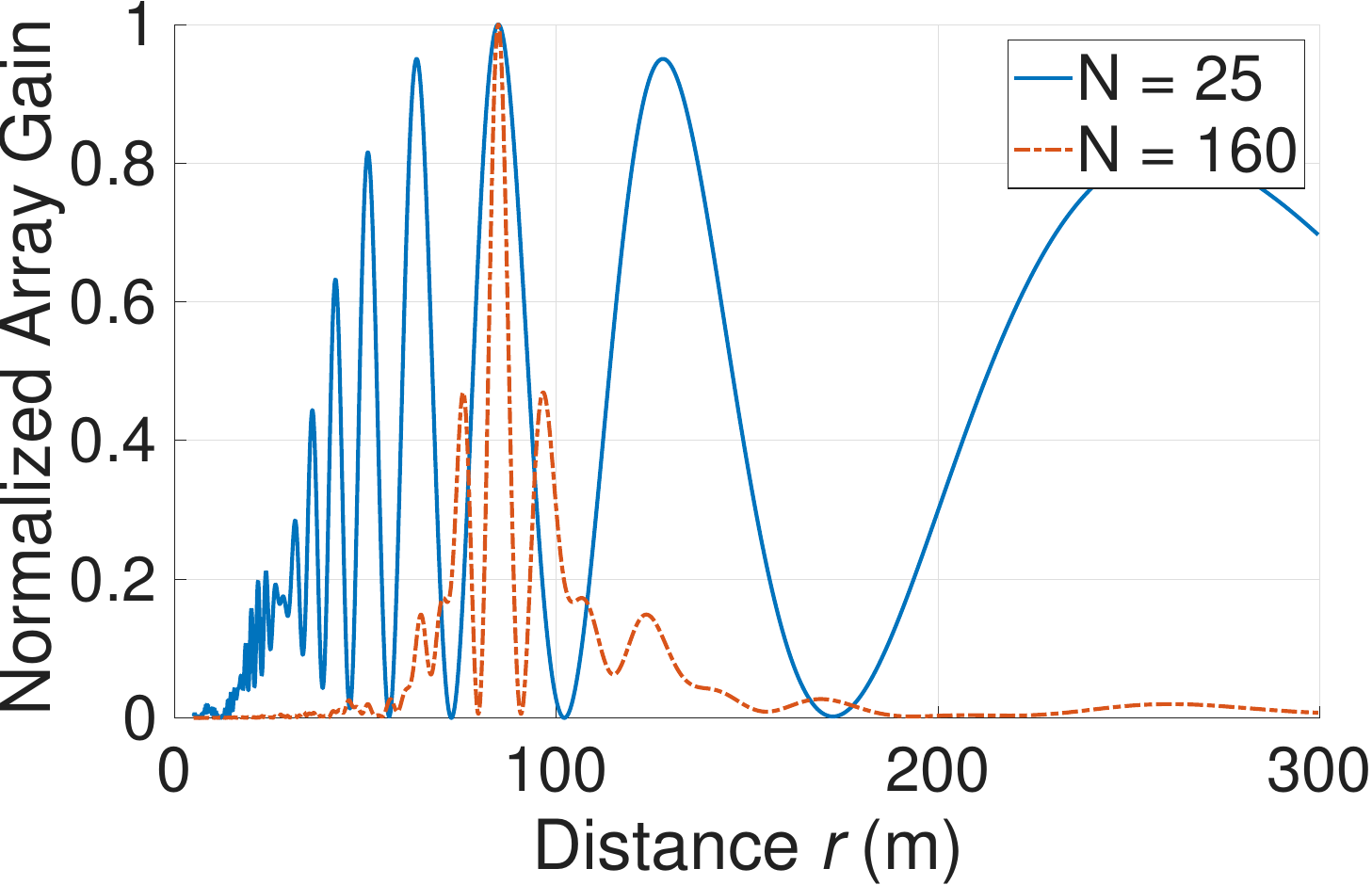}}
  \hfill
  \subfigure[]{
  \includegraphics[width=0.48 \linewidth]{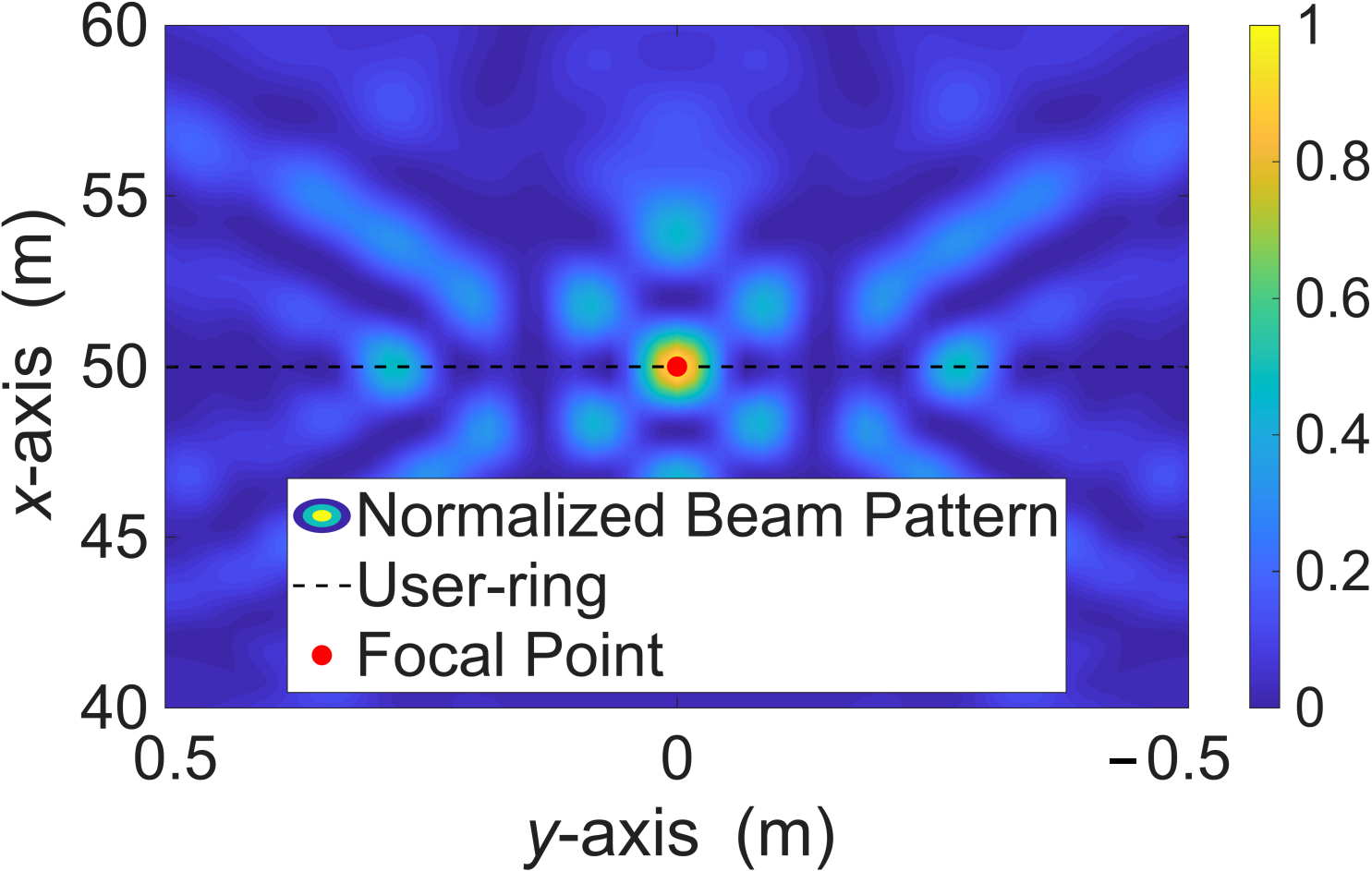}}
  \caption{Performance of increasing $N$ to suppress grating lobes for M-ULA when $L=$\ 4, $d_{\mathrm{z}}=$\ 2 m, $r_{\mathrm{F}}=$\ 50 m, $\theta_{\mathrm{F}}=$\ 0 and $f=$\ 15 GHz: (a) In the angle domain, $N=$\ 25; (b) In the angle domain, $N=$\ 160; (c) In the distance domain; (d) The beam pattern when $N=$\ 160.}
  \label{fig:ULA_suppression}
\end{figure}

We first consider the grating lobe suppression in the angle domain. As shown in (\ref{eq:G_BW}), the beam pattern in the user-ring could be decoupled into two Dirichlet kernels, with $\left|G_{N}(\theta)\right|^2$ representing the contribution of each single module, and $\left|G_{L}(\theta)\right|^2$ the arrangement of all modules. We notice that $G_N(\theta) = \frac{1}{N}\sum_{n=0}^{N-1} e^{j\kappa y_n\Delta}$ could be rewritten as $\operatorname{sinc}\left( \frac{N\Delta}{2} \right)$, where $\operatorname{sinc}(x) = \frac{\sin(\pi x)}{\pi x}$. Solving $\left| \operatorname{sinc}\left( \frac{N\Delta}{2} \right) \right|^2=$\ 0.5 yields $\frac{N\Delta}{2}=$\ 0.443. When $\theta_{\mathrm{F}}=$\ 0, the half-power beamwidth for the envelope is given by 
\begin{equation}
  \label{eq:BW_enve}
  \mathrm{BW}_\mathrm{enve} = \frac{1.772}{N}.
\end{equation}
According to (\ref{eq:location_even}) and (\ref{eq:location_odd_1}), in the user-ring, whether the number of modules $L$ is even or odd, the angle of the first grating lobe is always $\sin \theta_\mathrm{1} = \frac{\lambda}{d_{\mathrm{s}}}$. Based on the discussion above, the single-region beamfocusing condition could be expressed as
\begin{equation}
  \label{eq:beamfocusing_condition}
  \sin\theta_\mathrm{1} \geq \mathrm{BW}_\mathrm{enve}.
\end{equation}
By substituting $d_{\mathrm{s}}=d_{\mathrm{z}}+(N-1)d$, the single-region beamfocusing condition could be simplified as
\begin{equation}
  \label{eq:beamfocusing_condition_final}
  N \geq 0.795\left( \frac{d_{\mathrm{z}}}{d}-1 \right),
\end{equation}
where $d = \frac{\lambda}{2}$. (\ref{eq:beamfocusing_condition_final}) reveals the constraint between the number of antennas within each module $N$ and the spacing between two nearest antennas among two adjacent modules $d_\mathrm{z}$. To achieve single-region beamfocusing, the arrangement of modules cannot be too sparse. When $d_\mathrm{z}$ is large, i.e., when adjacent modules are far apart, the number of antennas within each module $N$ must increase correspondingly to make the main lobe of the envelope sufficiently narrow. 



Increasing $N$ will not only narrow the envelope in the angle domain, but also suppress the near-field grating lobes in the distance domain, as shown in Fig.~\ref{fig:ULA_suppression} (a), (b) and (c), respectively. According to (\ref{eq:G_C_l}), $C_l$ is a Dirichlet kernel in the form of a \textit{sinc} function, where the variable is $\frac{N(\Delta+Y_l\phi)}{2}$. Increasing $N$ will narrow the main lobe of this kernel function and make $C_l$ more sensitive to the phase shift brought by $Y_l\phi$, resulting in a rapid attenuation of array gain in the distance domain. Therefore, enlarging $N$ can suppress near-field grating lobes in both the angle and distance domains, as shown in Fig.~\ref{fig:ULA_suppression} (d).


The analysis above is based on $\theta_\mathrm{F}=$\ 0 and the assumption that $\Delta$ is small enough for approximation. In practical engineering, $\theta_\mathrm{F}=$\ 0 is typically used as a benchmark for design. For non-zero $\theta_\mathrm{F}$, the only difference is that we should consider $\Delta = \sin \theta - \sin \theta_\mathrm{F}$ and the corresponding first grating lobe location. The single-region beamfocusing condition in (\ref{eq:beamfocusing_condition}) remains the same.


The analysis in this subsection indicates that for M-ULA, simply increasing module spacing to expand the array aperture could not achieve the expected near-field beamfocusing precisely enough. To obtain a narrow envelope, the number of antennas per module should increase linearly with the module spacing, otherwise the grating lobe interferences will be severe. In the next subsection, an array structure named modular minimum-redundancy array (M-MRA) is proposed through redesigning the subarray configuration, which can suppress near-field grating lobes in both the angle and distance domains without increasing the number of antennas.

\subsection{Proposed M-MRA}

\begin{figure}[t!]
  \centering 
  \subfigure[]{
  \includegraphics[width=0.46 \linewidth]{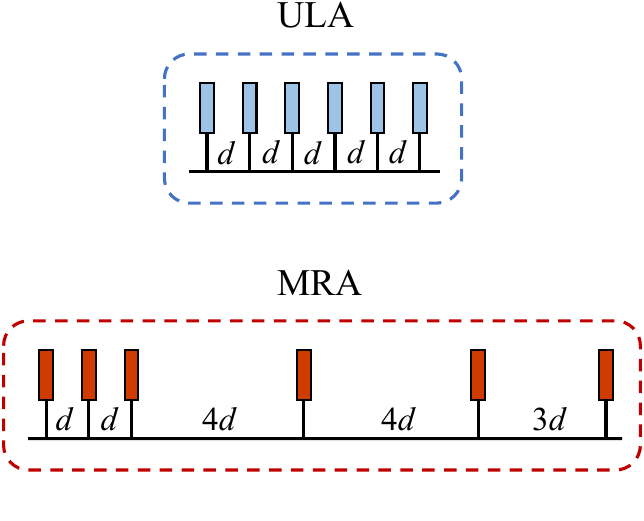}}
  \hfill
  \subfigure[]{
  \raisebox{0.05\height}{\includegraphics[width=0.49 \linewidth]{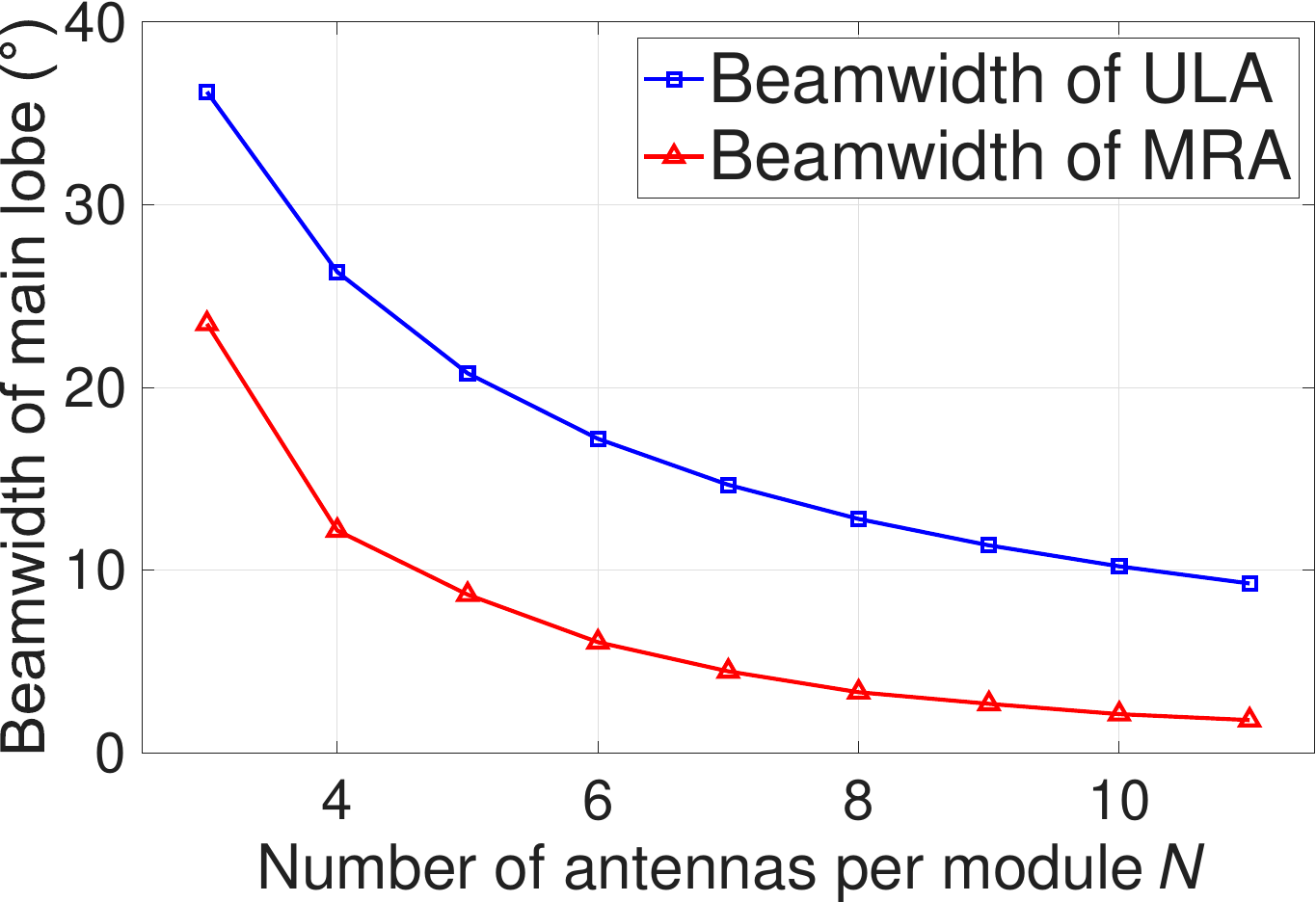}}}
  \caption{Illustration of MRA and ULA: (a) Configuration; (b) Beamwidth of main lobe.}
  \label{fig:ULA-MRA-value}
\end{figure}

The concept of minimum-redundancy array (MRA) was introduced in \cite{MRA_1968}. For a fixed number of antennas, the antenna locations are optimized to minimize the repetition of all possible distances between any two antennas, thereby expanding the array aperture and improving spatial resolution. 

For example, for an $N$-antenna ULA with the antenna spacing $d=\frac{\lambda}{2}$, its total aperture is $D_\mathrm{ULA}=Nd$. We can observe that a spacing of length $kd$ appears $(N-k)$ times, which wastes aperture resources and limits the spatial resolution of the array. In comparison, MRA minimizes the occurrence of spacings of different lengths by arranging antennas in a non-uniform manner. Specifically, when $N\leq$\ 4, a zero-redundancy array could be constructed, where all spacings from $d$ to $\frac{N(N-1)}{2}d$ each appear exactly once \cite{MRA_1968}, and the total aperture of MRA could reach its theoretical maximum $D_\mathrm{MRA}^\mathrm{zero}=\left(\frac{N(N-1)}{2}+1\right)d$, which is larger by a factor of about $\frac{N}{\rm{2}}$ compared to ULA. When $N>$\ 4, a zero-redundancy array cannot be achieved, but the arrangement of antennas could be optimized to minimize redundancy. \cite{MRA_1968} has given the configuration for MRA when $N\leq$\ 11. For larger $N$, \cite{MRA_N30} has given the configuration for MRA when $N\leq$\ 30, among which the largest aperture reaches $288d$, equal to 2.88 m at $f=$\ 15 GHz, sufficient to meet the size requirements of a single module for most practical array deployments. \cite{MRA_1968} also pointed out that for $N\geq$\ 20, the total aperture of MRA could be approximated as $D_\mathrm{MRA}^\mathrm{non-zero} \approx \frac{\rm{3}}{\rm{8}}N^2d$, which is larger by a factor of $\frac{\rm{3}}{\rm{8}}N$ compared to ULA. 

Fig.~\ref{fig:ULA-MRA-value} illustrates the array configuration for MRA and ULA when $N=$\ 6, as well as the beamwidth of main lobe for MRA and ULA under different $N$. Based on the analysis above, when $N \leq$\ 11, the main lobe beamwidth of an MRA is approximately $\frac{2}{N}$ times smaller than that of a ULA. 
For $N\geq$\ 20, the main lobe beamwidth of an MRA is approximately $\frac{8}{3N}$ times smaller than that of a ULA. It can be seen that MRA can achieve a larger array aperture and a narrower main lobe compared to ULA, thus possessing better spatial resolution and near-field beamfocusing capability. \cite{MRA_1995} has also revealed that MRA has deeper nulls for interferences when it is close to the desired signal direction.

\begin{figure}[t!]
  \centering 
  \subfigure[]{
  \includegraphics[width=0.47 \linewidth]{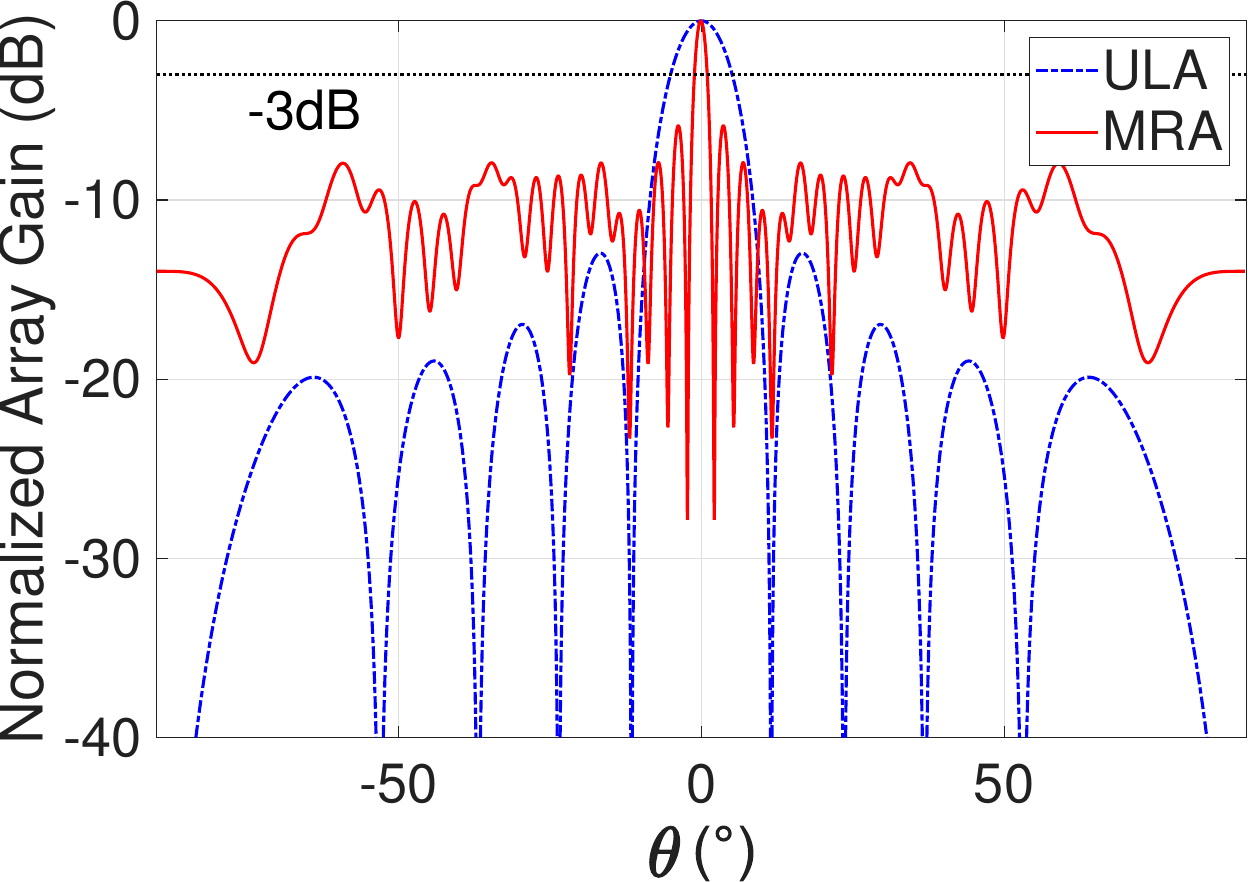}}
  \hfill
  \subfigure[]{
  \includegraphics[width=0.48 \linewidth]{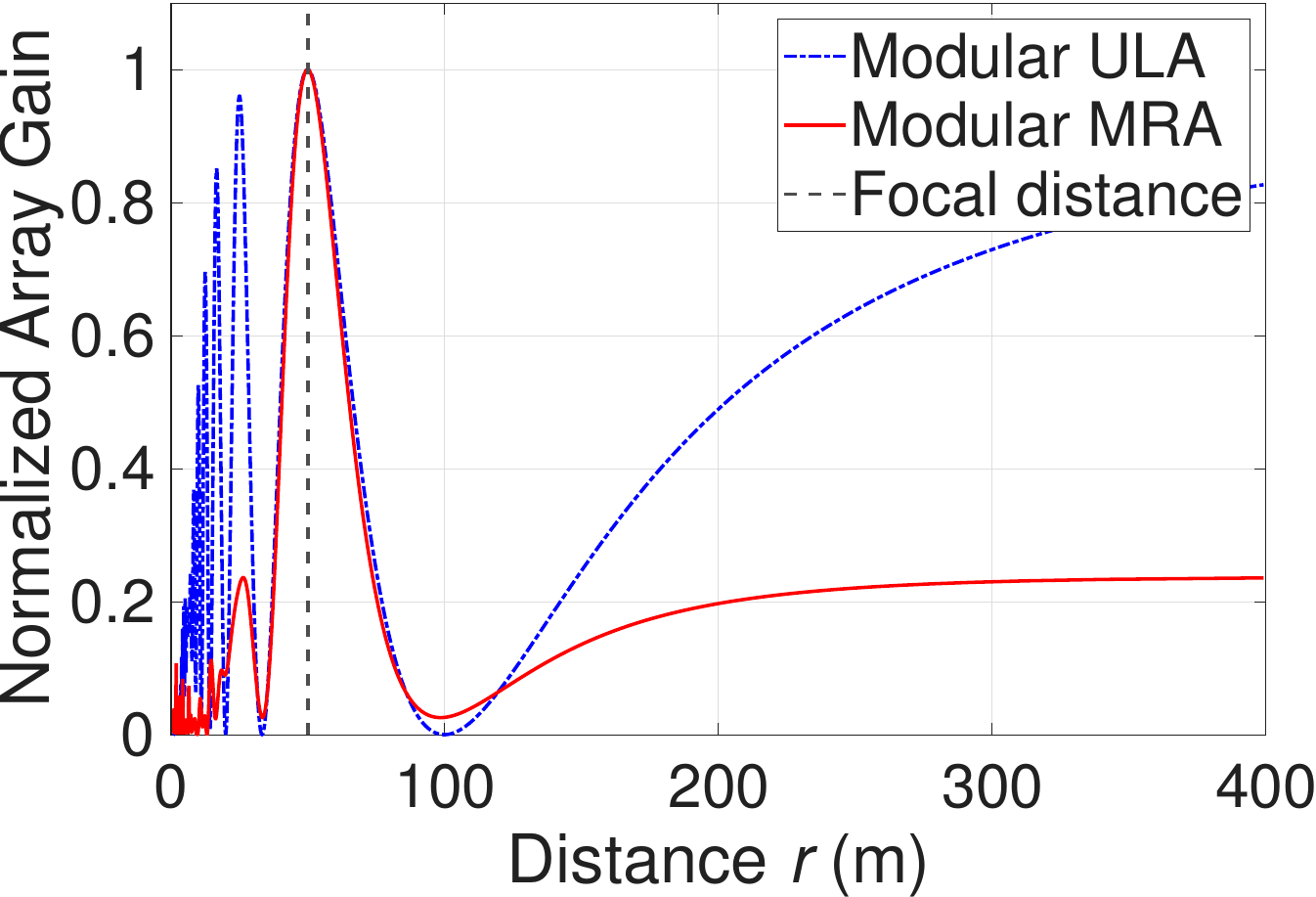}}
  \caption{Envelope comparison between ULA and MRA when $N=$\ 10: (a) in the angle domain; (b) in the distance domain.}
  \label{fig:MRA_suppression}
\end{figure}

\begin{figure}[t!]
  \centering 
  \subfigure[]{
  \includegraphics[width=0.48 \linewidth]{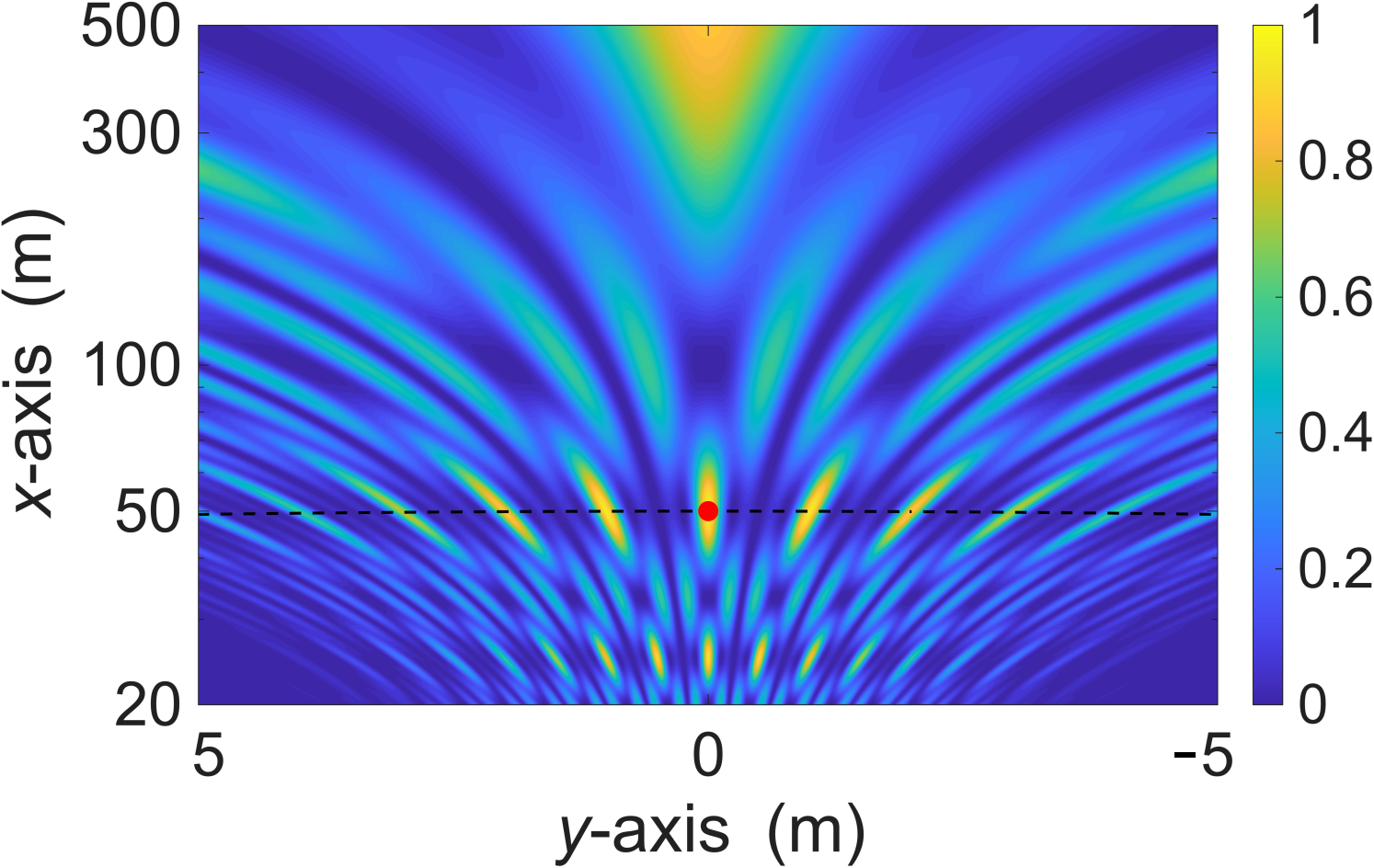}}
  \hfill
  \subfigure[]{
  \includegraphics[width=0.48 \linewidth]{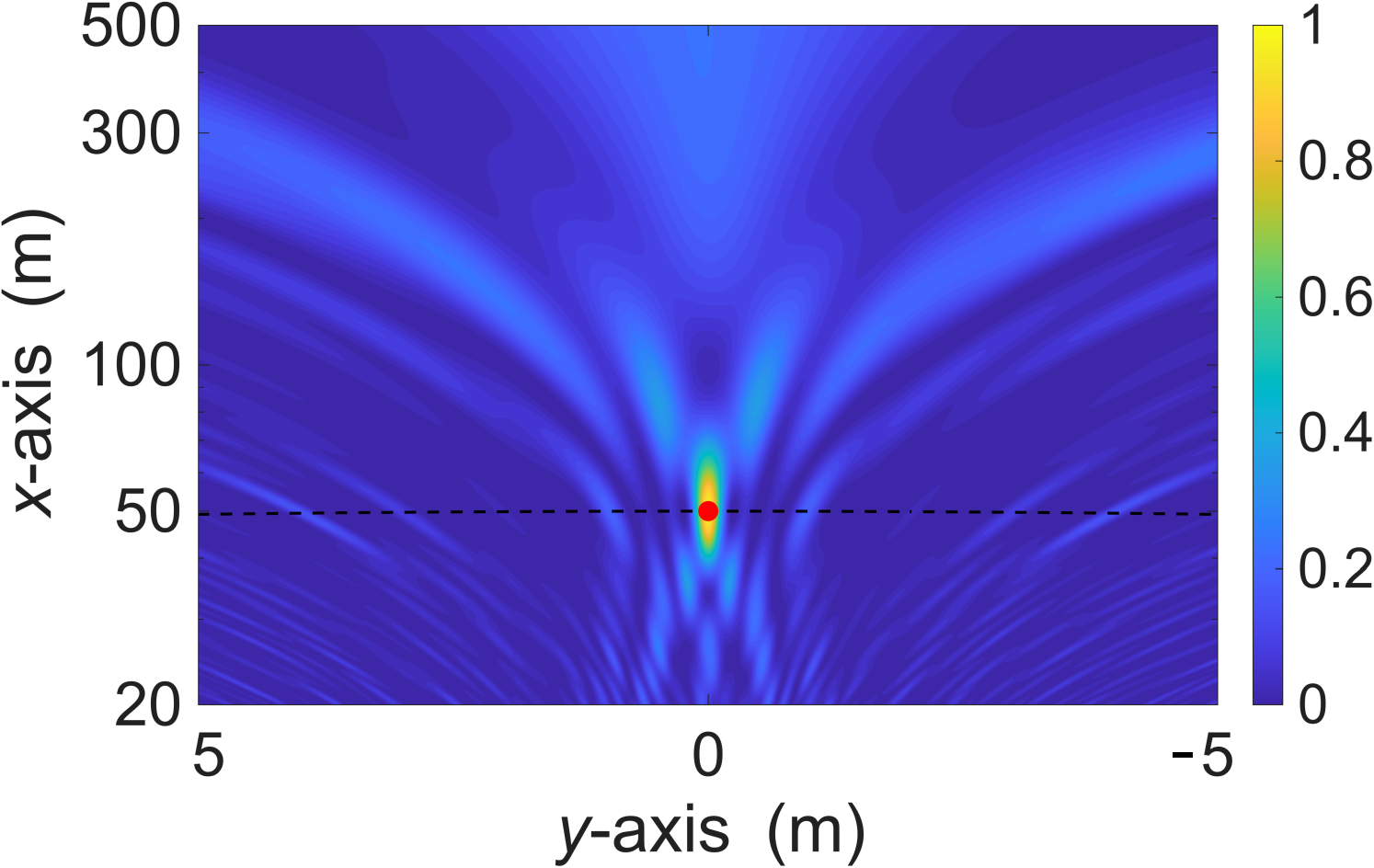}}
  \caption{Beam pattern comparison when $N=$\ 10, $L=$\ 4, $d_\mathrm{s}=$\ 1 m and $f=$\ 15 GHz: (a) M-ULA; (b) M-MRA.}
  \label{fig:ULA-MRA-pattern}
\end{figure}


In Section~\ref{section_analysis}, we have analyzed that for modular arrays, the slowly-varying envelope mainly depends on the structure of each single module. Based on the narrow main lobe characteristic of MRA, we combine MRA with the modular array, and propose the M-MRA structure to suppress near-field grating lobes. For M-MRA, each module is arranged as an MRA, whose narrow main lobe could ensure a narrow envelope in both the angle and distance domains. In this way, M-MRA can suppress near-field grating lobes in the two domains, as shown in Fig.~\ref{fig:MRA_suppression} (a) and Fig.~\ref{fig:MRA_suppression} (b) respectively. It should be noted that the sidelobe level of an MRA is generally higher than that of a ULA, which may introduce additional interferences in non-focusing regions. However, near-field beamfocusing is primarily concerned with the main lobe and its vicinity, and by applying amplitude weighting such as Hamming or Taylor windows to the outputs of each baseline of the MRA in the digital backend, the sidelobes of MRA can be suppressed to even lower levels \cite{MRA_1968}. Fig.~\ref{fig:ULA-MRA-pattern} shows the beam pattern under the same array parameters, which verifies the effectiveness of M-MRA in near-field grating lobe suppression.



In the following section, we will simulate and verify the grating lobe suppression capability of proposed M-MRA structure, as well as its improvement of the spectrum efficiency of multi-user near-field communications.

\section{Simulation Results}
\label{sim}

In this section, simulations are carried out to verify the effectiveness of deploying proposed M-MRA to suppress near-field grating lobes and enhance spectrum efficiency. Specifically, a downlink modular array communication system introduced in Section \ref{section system model} is considered. 
The carrier frequency $f$ is 15 GHz, and the corresponding unit of antenna spacing is $d=\frac{\lambda}{2}=$\ 1 cm. The beam patterns of M-ULA and M-MRA in multi-user near-field communications are first provided. Then the performance of average spectrum efficiency for six forms of arrays are compared and analyzed.

\subsection{Multi-user Near-Field Beamfocusing}

In this subsection, we set the number of users $K=$\ 4, which is a typical value in current communication systems. The number of modules is $L=$\ 4, the number of antennas per module is $N=$\ 10 and the inter-module spacing is $d_\mathrm{s}=$\ 1 m. The four users are at $(y,x)=$\ (0, 30 m), (0, 50 m), (-2 m, 55 m) and (3 m, 60 m), respectively. All users are located within the near-field region of the modular array, which aligns with typical urban micro-cell deployment scenarios. The user at (3 m, 60 m) falls outside the first grating lobe region of the user at (0, 50 m), while the user at (-2 m, 55 m) is deliberately placed within the grating lobe region of the user at (0, 50 m). 
The beam patterns for M-ULA and M-MRA are shown in Fig.~\ref{fig:multi pattern}, from which we can see that, based on the maximum ratio transmission (MRT), the total beam pattern is the superposition of $K$ sets of beam patterns focused on the $K$ users. 

We first start with the user at (3 m, 60 m) and the user at (0, 50 m). For M-ULA, since the grating lobe range of each set of beam pattern is rather large, there exist strong interferences among the patterns after superposition, making it difficult to achieve expected near-field beamfocusing for all users at the same time. Although the precoding vector for the user at (3 m, 60 m) allows a focusing region at (3 m, 60 m), the beam patterns after superposition prevent the formation of the expected high-gain region here. In contrast, proposed M-MRA achieves significantly improved near-field beamfocusing for each user, mainly because the MRA forms a narrower envelope in both the angle and distance domains, effectively preventing each set of beam pattern from overlapping. Even with a relatively small $N=$\ 10, the M-MRA still forms high-precision beamfocusing regions at each user. 

\begin{figure}[t!]
  \centering 
  \subfigure[]{
  \includegraphics[width=0.48 \linewidth]{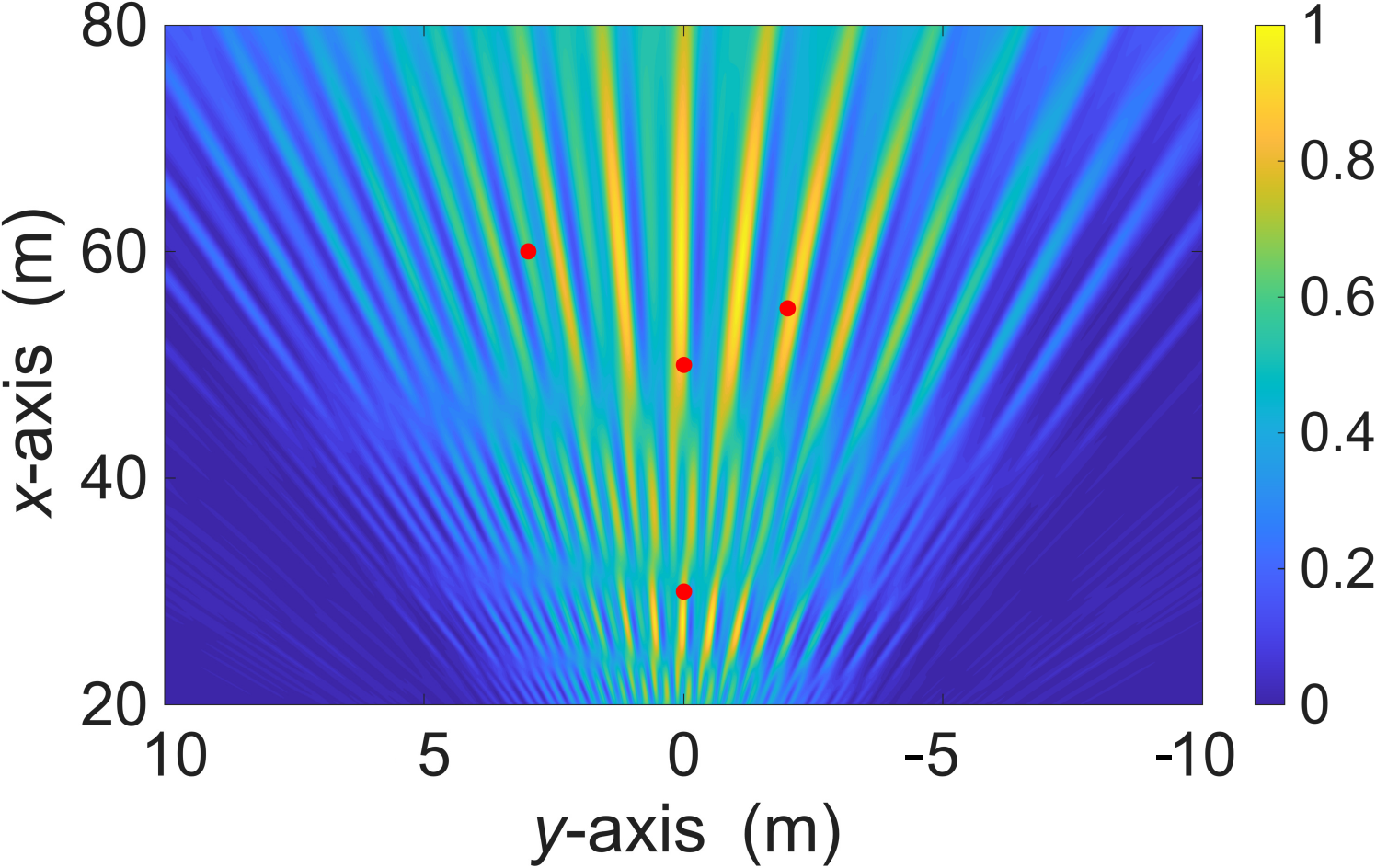}}
  \hfill
  \subfigure[]{
  \includegraphics[width=0.48 \linewidth]{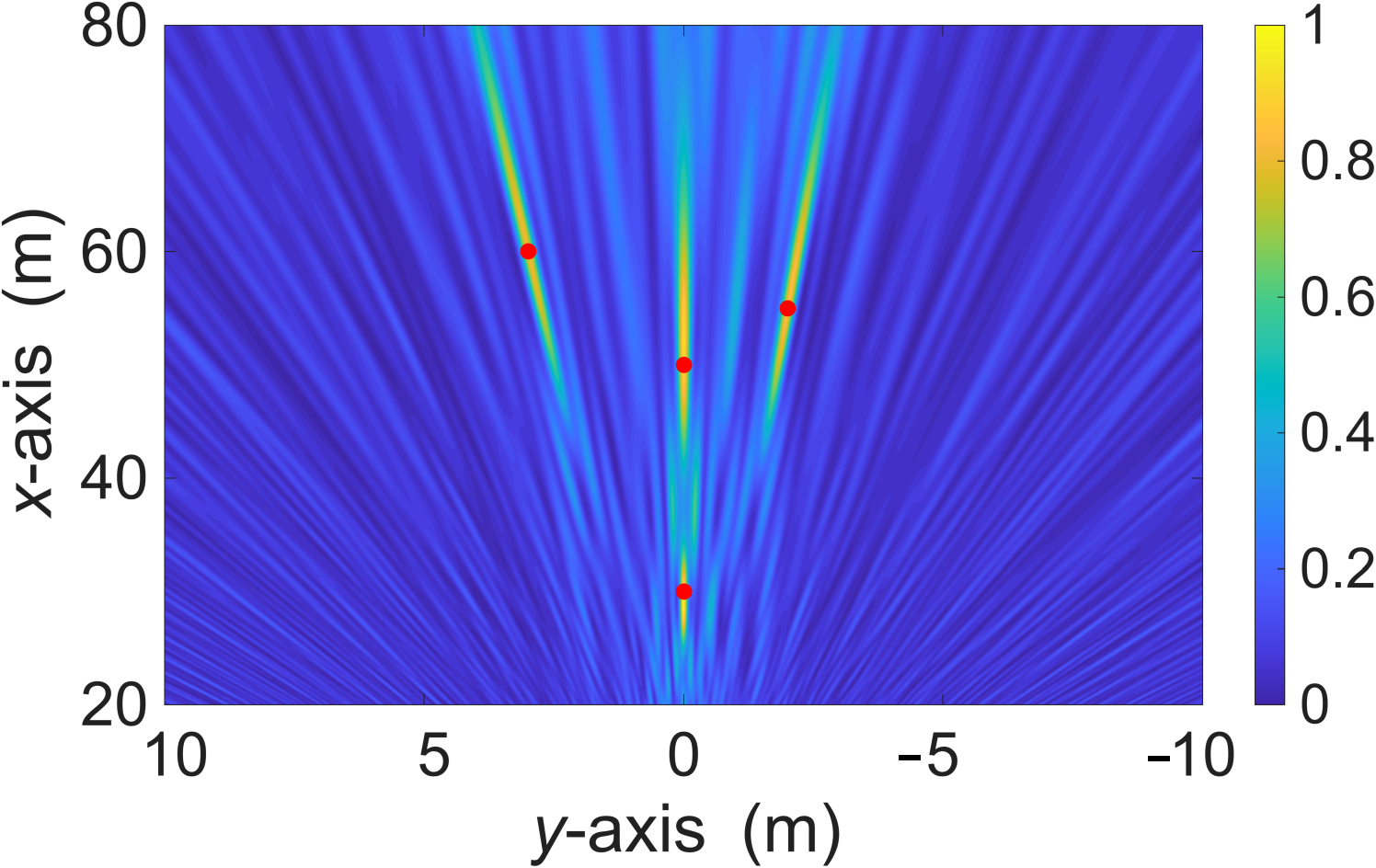}}
  \caption{Multi-user beam pattern comparison when $N=$\ 10, $L=$\ 4, $d_\mathrm{s}=$\ 1 m and $f=$\ 15 GHz: (a) M-ULA; (b) M-MRA.}
  \label{fig:multi pattern}
\end{figure}

We further pay attention to the user at (-2 m, 55 m), which is in the grating lobe region of the user at (0, 50 m). For M-ULA, though the total array gain at (-2 m, 55 m) is high, a significant portion of the energy should have been received by the user at (0, 50 m), causing non-negligible inter-user interferences. By comparison, proposed M-MRA reduces the overlap of high-array-gain regions among different users, achieving near-field beamfocusing for each user with higher precision. 



\subsection{Multi-user Average Spectrum Efficiency}

In this subsection, simulation results are provided to verify that M-MRA significantly improves the average spectrum efficiency upper bound for multi-user near-field communications. The signal model has been introduced in Section \ref{section system model}. The following types of arrays are considered for comparison:

\begin{itemize}

\item \textbf{M-ULA}: We set that the number of modules is $L=$\ 4, the number of antennas per module is $N=$\ 10. Modules are arranged uniformly with the adjacent module center spacing being $d_\mathrm{s}=$\ 1m. The total array aperture is $D=$\ 3.1 m, and the total number of antennas is $Q_\mathrm{M-ULA}=$\ 40. Antennas within each module are arranged as a ULA, as proposed in \cite{kosasih2026modular}.

\item\textbf{Proposed M-MRA}: We keep $L$, $N$, $d_\mathrm{s}$ and the total number of antennas the same as the M-ULA. Each module are designed as an MRA.


\item\textbf{TAS M-ULA}: We keep $L$, $N$ and the total number of antennas the same as the M-ULA. The only difference is that each module is distributed non-uniformly using the tapered aperture sampling (TAS) proposed in \cite{parvini2026tapered}.

\item\textbf{ECA M-ULA}: We replace each single antenna in the extended coprime array (ECA) proposed in \cite{zhou2025sparse} with a subarray. For the following simulation, parameters are designed to ensure that the total number of antennas and the array aperture are as close as possible to that of the M-ULA.

\item \textbf{Upper-bound ULA}: A collocated ULA with the same total array aperture as the M-ULA, requiring $Q_\mathrm{U-ULA}=$\ 310 antennas in total. Within the same array aperture, its antenna arrangement is the most dense, providing the most comprehensive spatial sampling and boasting the highest spatial resolution.

\item \textbf{Lower-bound ULA}: A collocated ULA with the same number of antennas as the M-ULA, requiring $Q_\mathrm{L-ULA}=$\ 40 antennas in total. Under the same number of antennas, its array aperture is the smallest, making its near-field range and capability of near-field beamfocusing the most limited in most cases.

\end{itemize}

\begin{figure}[t!]
	\centering
	\includegraphics[width=0.47\textwidth]{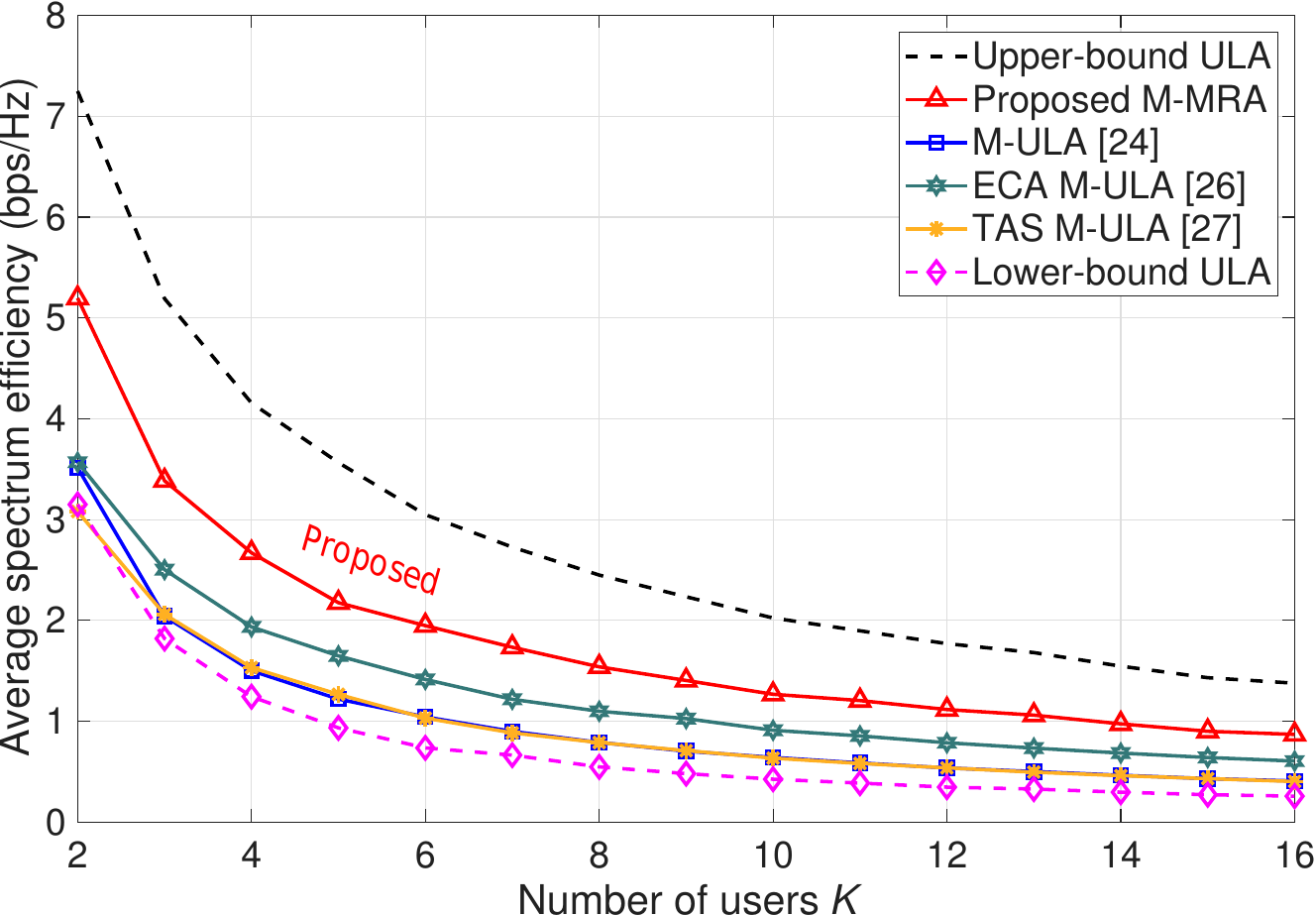}
	\caption{Average spectrum efficiency of five types of arrays against the number of users $K$.}
	\label{fig:SE angle uniform}
\end{figure}

Fig.~\ref{fig:SE angle uniform} shows the variation of the average spectrum efficiency against the number of users $K$ when users are uniformly distributed in the angular range [-3$^{\circ}$, 3$^{\circ}$] and distance range 
[20, 100] m. In this scenario, the distribution of user angles is relatively concentrated, so that the effectiveness in resisting inter-user interferences primarily relies on the array's resolution in the angle domain. Due to the narrow envelope characteristic, proposed M-MRA outperforms other modular arrays. Taking a case with $K=$\ 4 as an example, the average spectrum efficiency of the M-ULA is approximately 1.5 bps/Hz. With the same number of antennas, the average spectrum efficiency of proposed M-MRA is approximately 2.7 bps/Hz, representing an improvement of about 80\%.


It should be pointed out that MRA reduces the width of main lobe at the cost of increasing the sidelobe height, as shown in Fig.~\ref{fig:MRA_suppression} (a). Therefore, as the angle range of user random distribution increases, the performance improvement brought by M-MRA will decrease. As demonstrated in Fig.~\ref{fig:SE angle range}, 
when the half-angle range of the random distribution of users reaches about 30$^{\circ}$, the average spectrum efficiency that M-MRA can achieve is no higher than that of M-ULA. However, this will not limit the use of M-MRA in the near-field beamfocusing. Taking $N=$\ 10 as an example, the half-angle width of the main lobe of ULA is about 5$^{\circ}$. Since the near-field beamfocusing mainly acts within the main lobe beamwidth, as long as the M-MRA has a significant performance improvement within about [-5$^{\circ}$, 5$^{\circ}$], it will bring corresponding near-field beamfocusing gain.


\begin{figure}[t!]
	\centering
	\includegraphics[width=0.47\textwidth]{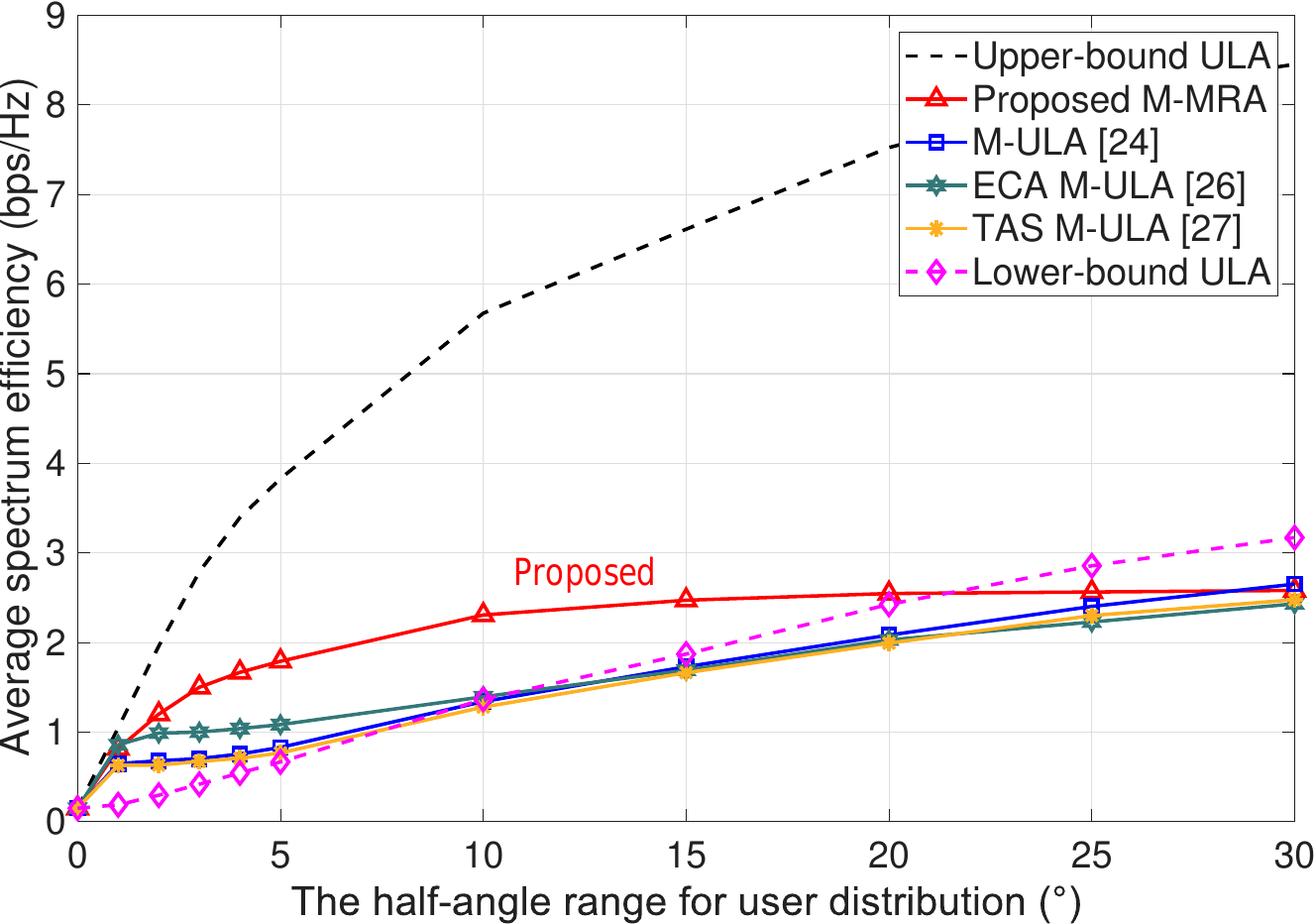}
	\caption{Average spectrum efficiency of five types of arrays against the half-angle range for user distribution.}
	\label{fig:SE angle range}
\end{figure}

\begin{figure}[t!]
	\centering
	\includegraphics[width=0.47\textwidth]{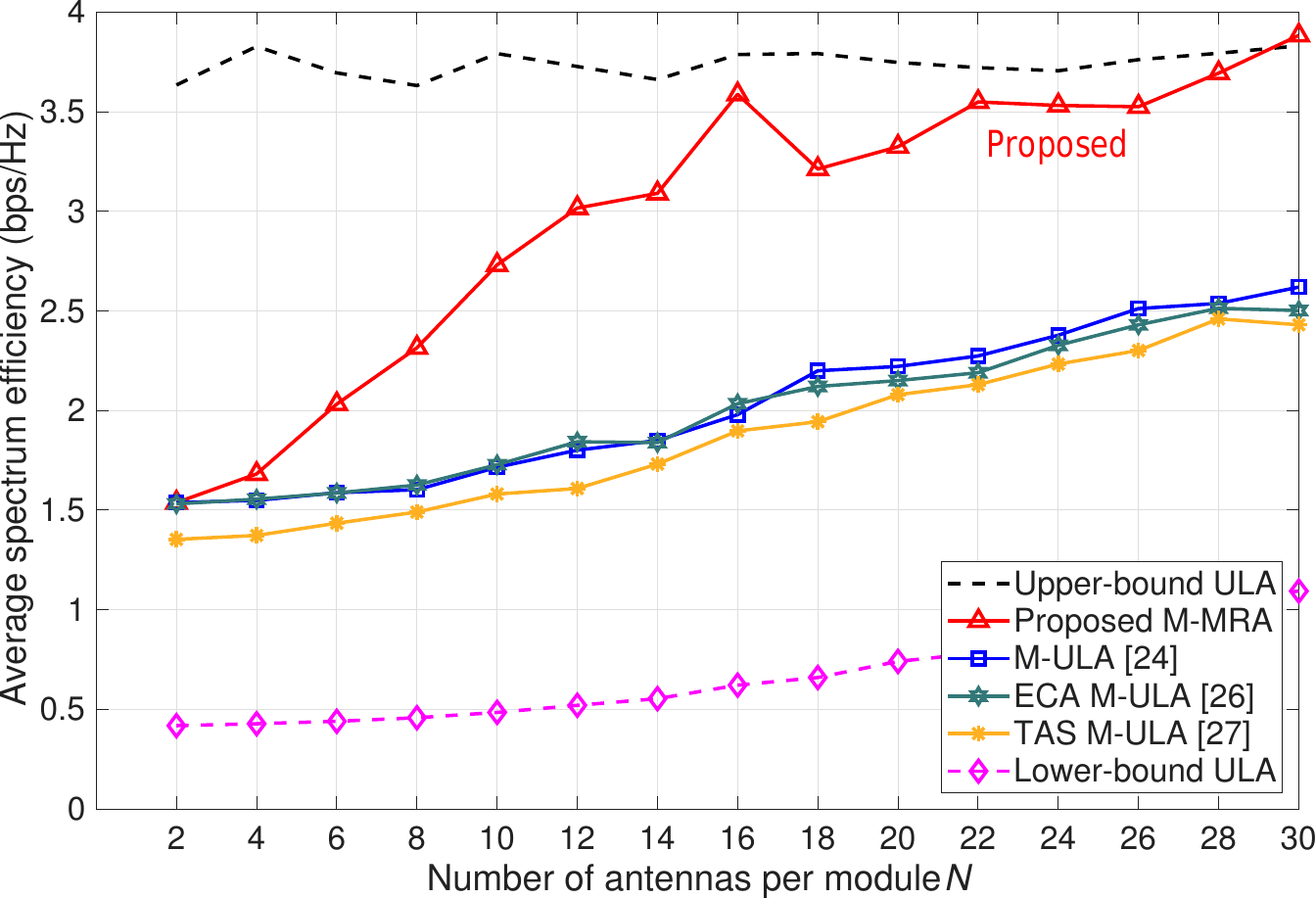}
	\caption{Average spectrum efficiency of five types of arrays against the number of antennas per module $N$.}
	\label{fig:SE N}
\end{figure}

Fig.~\ref{fig:SE N} illustrates the average spectrum efficiency versus the number of antennas per module $N$. When the number of modules is limited, the ECA M-ULA degenerates into sparse array. The TAS M-ULA even performs worse than M-ULA, for the non-uniform sampling requires a sufficient number of modules. Besides, the curve labeled as the upper-bound ULA remains approximately flat, for its total number of antennas is almost unchanged. The lower-bound ULA provides the lowest spectrum efficiency over the entire range of $N$, for its compact physical aperture limits the near-field beamfocusing accuracy and spatial resolution. As $N$ increases, the effectiveness of spectrum efficiency enhancement for proposed M-MRA is more evident. An M-MRA with $N=$\ 8 and an M-ULA with $N=$\ 24 exhibit comparable performance, when the M-MRA only uses about 33\% of the number of antennas.

\begin{figure}[t!]
	\centering
	\includegraphics[width=0.45\textwidth]{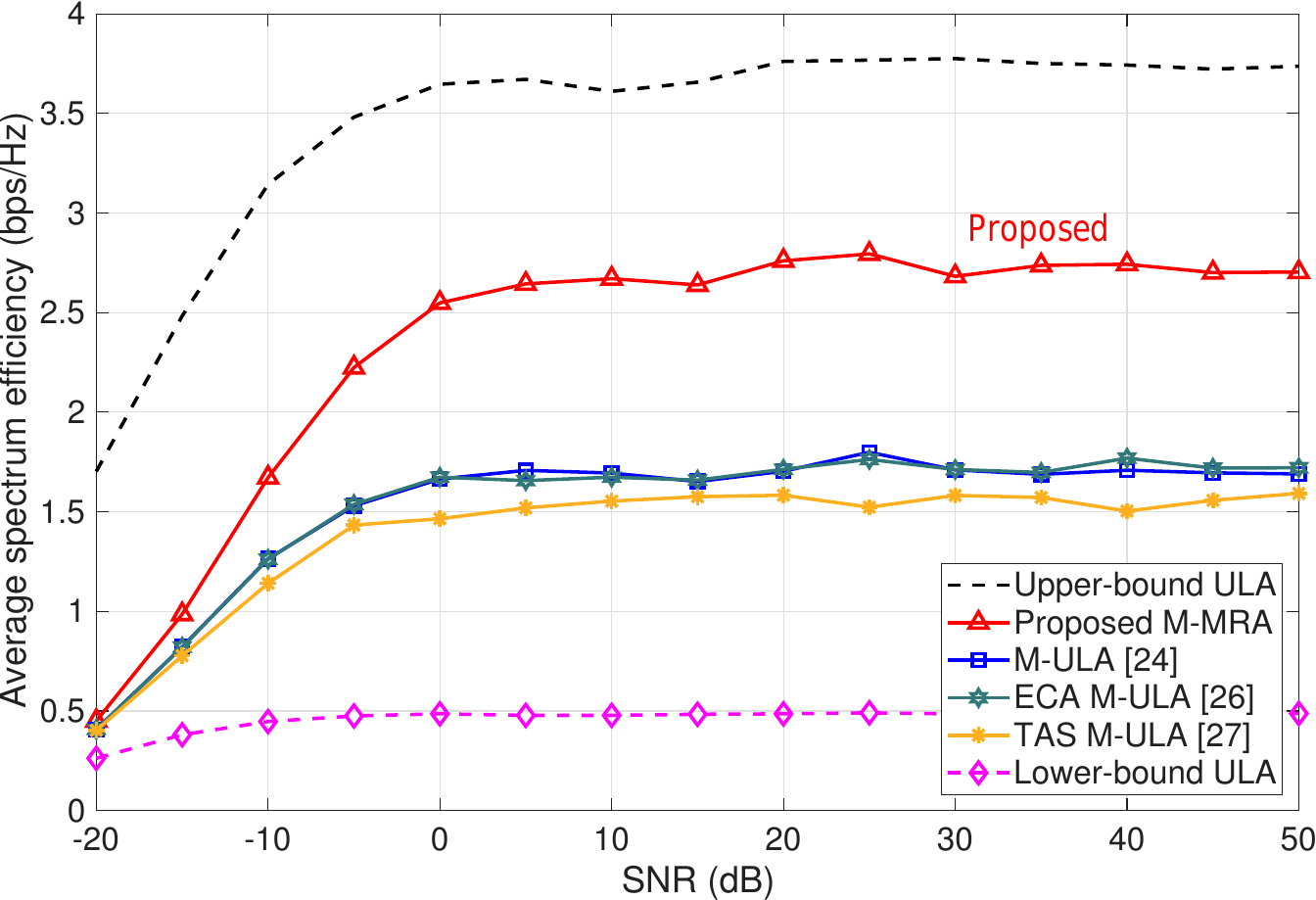}
	\caption{Average spectrum efficiency of five types of arrays against SNR.}
	\label{fig:SE SNR}
\end{figure}

 The average spectrum efficiency versus SNR is presented in Fig.~\ref{fig:SE SNR}. The spectrum efficiency of all array configurations increases with the SNR and gradually saturates in the high-SNR regime, indicating a transition from a noise-limited system to an interference-limited system. In the low-SNR regime, approximately from -20 to -10 dB, the receiver noise is the dominant performance-limiting factor. Therefore, increasing the transmit SNR directly improves the user SINRs, resulting in a rapid increase in the average spectrum efficiency. 
 When the SNR exceeds approximately 10 dB, all curves tend to saturate. In this regime, the noise term becomes negligible, and the user SINR is mainly determined by the ratio between the desired-signal power and the residual inter-user interferences. These results demonstrate that, in high-SNR multi-user near-field systems, improving the spatial resolution to suppress grating lobe interferences is more effective than simply increasing the transmit power.


We have mentioned in Section \ref{section_analysis} that, as $d_{\mathrm{s}}$ increases, the beamdepth of the main lobe decreases, which improves the resolution in the distance domain, but more grating lobes emerge, which causes stronger interferences. The merits and demerits of increasing $d_{\mathrm{s}}$ coexist, so that for modular arrays, there is likely to be an inflection point regarding the uniform inter-module spacing $d_{\mathrm{s}}$, where the effects of the two can be optimally balanced. Fig.~\ref{fig:SE ds} illustrates the existence of inflection points regarding $d_{\mathrm{s}}$. This observation corroborates our analysis in Section \ref{section_analysis} and provides a guideline for modular array design that, the inter-module spacing should be chosen carefully rather than maximized blindly to enlarge the aperture.

\section{Conclusions}

In this paper, we first reveal the existence of near-field grating lobes in the distance domain based on the analysis of near-field beam patterns for modular arrays. 
We further explore how to suppress near-field grating lobes for the M-ULA. We find that increasing the number of antennas per module can suppress grating lobes, but the grating lobe interferences are severe under a limited number of antennas. This limitation inspires us to propose the M-MRA structure by redesigning the subarray configuration. Simulation results verify that the proposed M-MRA can suppress near-field grating lobes and enhance the spectrum efficiency of multi-user near-field communications significantly. 
This work proposes an array configuration that balances near-field range expansion and grating lobe suppression, providing a feasible solution for deploying modular arrays in near-field communications. 
For future work, one may further explore the inflection point regarding the inter-module spacing, or combine subarray configuration designing with module position optimization \cite{Emil2026modulecenter}, or extend the analysis from linear arrays to planar arrays \cite{Emil2026twoD}, \cite{distanceUPA2026}.

\begin{figure}[t!]
	\centering
	\includegraphics[width=0.445\textwidth]{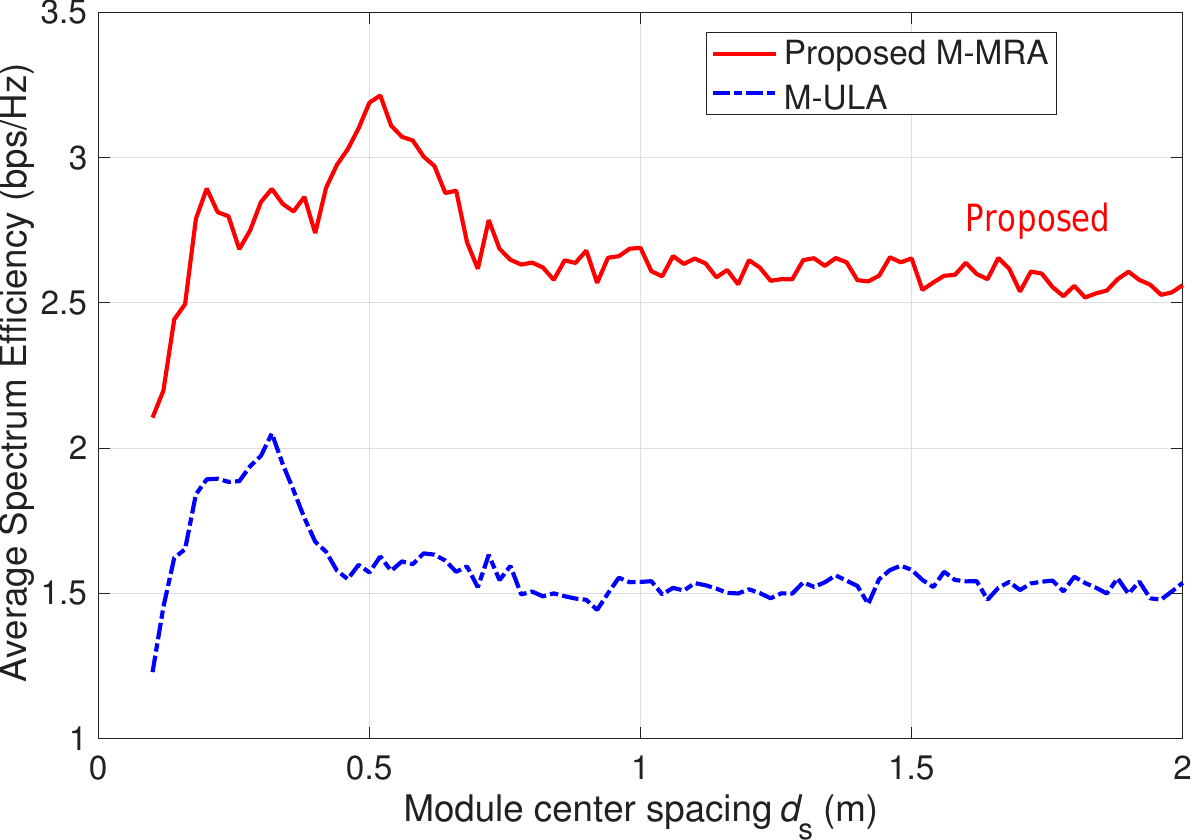}
	\caption{Average spectrum efficiency of three types of modular arrays against the module center spacing $d_\mathrm{s}$.}
	\label{fig:SE ds}
\end{figure}

\begin{appendices}

\section{Explanation of the distinctiveness of L\ =\ 2}
\label{distinctiveness_L_2}

We consider the case where the modular array consists of $L=$\ 2 modules placed symmetrically about the array center, thus the sub-array centers are located at $ Y_0 = -\frac{d_{\mathrm{s}}}{2}$ and $Y_1 = \frac{d_{\mathrm{s}}}{2}$. Substituting $Y_0$, $Y_1$ into (\ref{eq:G_C_l}), we have
\begin{equation}
    \label{eq:G_C_l_2}
    G(r, \theta) = \left| \frac{1}{2} \left( C_0 e^{-j\kappa \frac{d_{\mathrm{s}}}{2} \Delta} e^{j\kappa \frac{d_{\mathrm{s}}^2}{8} \phi} + C_1 e^{j\kappa \frac{d_{\mathrm{s}}}{2} \Delta} e^{j\kappa \frac{d_{\mathrm{s}}^2}{8} \phi} \right) \right|^2.
\end{equation}
Since both terms share the identical quadratic phase factor \(e^{j\kappa \frac{d_{\mathrm{s}}^2}{8} \phi}\), (\ref{eq:G_C_l_2}) can be simplified to 
\begin{equation}
    \label{eq:G_C_l_2_com}
    G = \frac{1}{4} \left| C_0 e^{-j\kappa \frac{d_{\mathrm{s}}}{2} \Delta} + C_1 e^{j\kappa \frac{d_{\mathrm{s}}}{2} \Delta} \right|^2.
\end{equation}
Based on the definition $C_l \triangleq \frac{1}{N} \sum_{n=0}^{N-1} e^{j\kappa y_n (\Delta + Y_l \phi)}$ in Section \ref{section_analysis}, for \(L=2\), we have
\begin{equation}
    C_0 = \frac{1}{N} \sum_{n} e^{j\kappa y_n (\Delta - \frac{d_{\mathrm{s}}}{2}\phi)}, \quad
    C_1 = \frac{1}{N} \sum_{n} e^{j\kappa y_n (\Delta + \frac{d_{\mathrm{s}}}{2}\phi)}.
\end{equation}
Since the intra-module antenna positions \(y_n\) are symmetric about the module center, both \(C_0\) and \(C_1\) are real-valued \textit{sinc} functions. Expanding (\ref{eq:G_C_l_2_com}) using \(|X+Y|^2 = |X|^2 + |Y|^2 + 2\Re\{XY^*\}\), and noting that \(C_0, C_1 \in \mathbb{R}\), we obtain
\begin{equation}
    \label{eq:G_C_l_2_final}
    G = \frac{1}{4} \left[ C_0^2 + C_1^2 + 2 C_0 C_1 \cos\left(\kappa d_{\mathrm{s}} \Delta\right) \right],
\end{equation}
which shows that the phase term determining the interference oscillation, \(\cos(\kappa d_{\mathrm{s}} \Delta)\), 
is independent of the curvature difference \(\phi\), and thus independent of the observation distance \(r\). Although \(C_0\) and \(C_1\) depend on \(\phi\) and may exhibit peaks, the amplitude of these peaks is limited by the product of $C_0 C_1$. Even if both $C_0$ and $C_1$ simultaneously reach their respective first grating lobes with an amplitude of about 0.2, the amplitude factor for $\cos\left(\kappa d_{\mathrm{s}} \Delta\right)$ is only about 0.04
. Consequently, for an M-ULA with \(L=2\), the distance variation of \(G\) is limited to the slowly-varying envelope contributed by the subarray factors \(C_0\) and \(C_1\), making the beam pattern not exhibit grating-lobe level peaks.

\end{appendices}

\footnotesize
\bibliographystyle{IEEEtran}
\bibliography{IEEEabrv,reference}

\vfill

\end{document}